\documentclass[a4paper,11pt]{article}

\usepackage{fullpage}
\usepackage{authblk}

\usepackage{times}  % DO NOT CHANGE THIS
\usepackage{helvet}  % DO NOT CHANGE THIS
\usepackage{courier}  % DO NOT CHANGE THIS
\usepackage[hyphens]{url}  % DO NOT CHANGE THIS
\usepackage{graphicx} % DO NOT CHANGE THIS
\usepackage{natbib}  % DO NOT CHANGE THIS AND DO NOT ADD ANY OPTIONS TO IT
\usepackage{caption} % DO NOT CHANGE THIS AND DO NOT ADD ANY OPTIONS TO IT
\DeclareCaptionStyle{ruled}{labelfont=normalfont,labelsep=colon,strut=off} % DO NOT CHANGE THIS
\usepackage{algorithm}
\usepackage{algorithmic}

\usepackage{newfloat}
\usepackage{listings}
\floatstyle{ruled}
\newfloat{listing}{tb}{lst}{}
\floatname{listing}{Listing}

\usepackage{amssymb,amsmath,amsfonts,amstext,amsthm}
\usepackage[inline]{enumitem}
\usepackage{booktabs, makecell}

\renewcommand{\cite}[1]{\citep{#1}}

\usepackage{pgfplots}
\usepackage{pgfplotstable}

\usepackage{tikz}
\usetikzlibrary{shapes.misc, automata, positioning, arrows, arrows.meta}

\newtheorem{lemma}{Lemma}
\newtheorem{proposition}[lemma]{Proposition}
\newtheorem{theorem}[lemma]{Theorem}
\newtheorem{corollary}[lemma]{Corollary}

\theoremstyle{definition}

\newtheorem{example_}[lemma]{Example}

\newcommand{\pinothing}{\perp}

\usepackage{etoolbox}

\newlength\figurewidth
\newlength\figureheight
\newtoggle{shortversion}
\settoggle{shortversion}{true}
\newenvironment{NoHyper}{}{}

\newenvironment{talign*}
 {\csname align*\endcsname}
 {\endalign}

\title{Bayesian Persuasion in Sequential Decision-Making\thanks{Work done when Jiarui Gan was affiliated with Max Planck Institute for Software Systems.}}

\author[1]{Jiarui Gan}
\author[2]{Rupak Majumdar}
\author[2]{Goran Radanovic}
\author[2]{Adish Singla}

\affil[1]{University of Oxford}
\affil[2]{Max Planck Institute for Software Systems}

\date{}

\begin{document}

\maketitle

\begin{abstract}
We study a dynamic model of Bayesian persuasion in sequential decision-making settings. An informed principal observes an  external parameter of the world and advises an uninformed agent about actions to take over time. The agent takes actions in each time step based on the current state, the principal's advice/signal, and beliefs about the external parameter. The action of the agent updates the state according to a stochastic process. The model arises naturally in many applications, e.g., an app (the principal) can advice the user (the agent) on possible choices between actions based on additional real-time information the app has.
We study the problem of designing a signaling strategy from the principal's point of view. We show that the principal has an optimal strategy against a myopic agent, who only optimizes their rewards locally, and the optimal strategy can be computed in polynomial time. In contrast, it is NP-hard to approximate an optimal policy against a far-sighted agent. Further, if the principal has the power to threaten the agent by not providing future signals, then we can efficiently compute a threat-based strategy. This strategy guarantees the principal's payoff as if playing against an agent who is far-sighted but myopic to future signals.
\end{abstract}

\section{Introduction}

Uncertainty is prevalent in models of sequential decision making. 
Usually, an agent relies on prior knowledge and Bayesian updates as a basic approach to dealing with uncertainties. 
In many scenarios, a knowledgeable \emph{principal} has direct access to external information and can reveal it to influence the agent's behavior.
For example, a navigation app (the principal) normally knows about the global traffic conditions and can inform a user (the agent), who then decides a particular route based on the app's advice.
The additional information can help improve the quality of the agent's decision-making. Meanwhile, by strategically revealing the external information, the principal can also persuade the agent to act in a way beneficial to the principal.

We study the related persuasion problem in a dynamic environment.
In a static setting, the interaction between the principal and the agent is modeled by \emph{Bayesian persuasion} \cite{kamenica2011bayesian}, where the principal uses their information advantage to influence the agent's strategy in a one-shot game, by way of signaling.
In this paper, we extend this setting to include interaction in an infinite-horizon Markov decision process (MDP), where rewards incurred depend on the state of the environment, the action performed, as well as an external parameter sampled from a known prior distribution at each step.
The principal, who cannot directly influence the state, observes the realization of this external parameter and signals the agent about their observation. 
The agent chooses to perform an action based on the state and the signal, and the action updates the state according to a stochastic transition function.
Both the principal and the agent aim to optimize their own rewards received in the course of the play.

If the objectives of the principal and the agent are completely aligned, the principal should reveal true information about the external parameter, so the more interesting case is when they are misaligned.
For example, a user of an navigation app only wants to optimize their commute times but the app may want to incentivize the user to upgrade to a better service, or to increase traffic throughput when the app is provided by a social planner.
We consider two major types of agents---\emph{myopic} and \emph{far-sighted}---and investigate the problem of optimal signaling strategy design against them.
A myopic agent optimizes their payoff locally: in each step,
they take an action that will give them the highest immediate reward. 
It can model a large number of ``short-lived'' agents each appearing instantly in a system 
(e.g., users of a ride-sharing app or an E-commerce website).
A far-sighted agent, on the other hand, optimizes their long-run cumulative reward and considers future information disclosure.

We show that, in the myopic setting, an optimal signaling strategy for the principal can be computed in polynomial time through a reduction to linear programming. 
On the other hand, in the case of a far-sighted agent, optimal signaling strategy design becomes computationally intractable: 
if P$\neq$NP, there exists no polynomial time approximation scheme.
Our proof of computational intractability is quite general, and extends to showing the hardness of similar principal-agent problems in dynamic settings.

To work around the computational barrier, we focus on a special type of far-sighted agents who are \emph{advice-myopic}. An advice-myopic agent optimizes their cumulative reward over time based on the history of information disclosures, but does not assume that the principal will continue to provide information in the future.
We expect such behavior to be a natural heuristic in the real world when agents resort to prior knowledge, but not future information disclosure, to estimate future rewards.
We then show that optimal signaling strategies can again be computed in polynomial-time.
More interestingly, the solution can be used to design a threat-based signaling strategy against a far-sighted agent. We show that this threat-based strategy induces the same reaction from a far-sighted agent as from an advice-myopic one. Hence, it guarantees the principal the same payoff obtained against an advice-myopic agent, when the agent is actually far-sighted.
Figure~\ref{fig:mdp-example} shows the subtleties of optimal signaling strategies in the dynamic setting.

\begin{figure}[t]
%\vspace{-3mm}
\centering
\tikzset{
->, % makes the edges directed
>={Stealth[scale=1.2]}, % makes the arrow heads bold
node distance=40, % specifies the minimum distance between two nodes. Change if necessary.
every state/.style={thick, fill=gray!10,minimum size=20}, % sets the properties for each ’state’ node
every edge/.append style={semithick},
%initial text= , % sets the text that appears on the start arrow
}
%\hspace{-5mm}
\begin{tikzpicture}[baseline={(current bounding box.center)}]
\tikzstyle{every node}=[font=\small] 

\node[state, initial, initial where=left, initial text = ] (s0) {$s_0$};

\node[state, right of=s0, xshift=20,yshift=40] (s1) {$s_1$};
\node[state, accepting, right of=s0, xshift=45] (s2) {$s_2$};

\draw 
(s0) edge[thick, dashed, bend left, above] node{$a$} (s1)
(s0) edge[thick, dashed, bend right, below] node{$b$} (s1)
(s1) edge[thick,out=140,in=100,looseness=1.3, above] node{$a$ $(0, 0)$} (s0)
(s1) edge[thick, bend left, right, looseness=1] node[pos=0.05]{$~b$}  node[pos=0.6]{$(0.1, 10)$} (s2)
(s0) edge[thick,out=-40,in=-140, looseness=0.9, below] node{$c$ $(0.1,0)$} (s2)
;
\end{tikzpicture}
\qquad
\begin{small}
\begin{tabular}{ crr } 
\toprule
  & \multicolumn{2}{c}{External parameter}\\
  & $\theta_a$ & $\theta_b$ \\
\midrule
$a$ & $(1, 1)$  & $(-1, 0)$ \\ 
$b$ & $(-1, 0)$ & $(1, 1)$ \\
\bottomrule
\end{tabular}
\end{small}
%\vspace{-2mm}
\caption{A simple example: a principal wishes to reach $s_2$ while maximizing rewards. 
All transitions are deterministic, and every edge is labeled with the corresponding action and (in the brackets) rewards for the agent and the principal, respectively. The rewards for state-action pairs $(s_0, a)$ and $(s_0, b)$ (dashed edges) also depend on the 2-valued external parameter, as specified in the table; the value of the parameter is sampled uniformly at random at each step. 
Assume uniform discounting with discount $\frac{1}{2}$. 
Without signaling, the agent will always take action $c$ in $s_0$, whereby the principal obtains payoff $0$. The principal can reveal information about the external parameter to attract the agent to move to $s_1$. If the agent is myopic, the principal can reveal full information, which leads to the agent moving to $s_1$, taking action $b$, and ending in $s_2$; the principal obtains payoff $6$ as a result. However, if the agent is far-sighted, this will not work: the agent will end up in a loop in $s_0$ and $s_1$, resulting in overall payoff $4/3$ for the principal. To improve, the principal can use a less informative strategy in $s_0$: e.g., advising the agent to take the more profitable action $10\%$ of the time and a uniformly sampled action in $\{a,b\}$ the remaining $90\%$ of the time. The agent will be incentivized to move to $s_1$ then. Alternatively, the principal can also use a threat-based strategy, which yields an even higher payoff in this instance: always reveal the true information in $s_0$, advise the agent to take $b$ in $s_1$, and stop providing any information if the agent does not follow the advice.
The outcome of this strategy coincides with how an advice-myopic agent behaves: they will choose $b$ at $s_1$ as future disclosures are not considered.
}
%\vspace{-3mm}
\label{fig:mdp-example}
\end{figure}
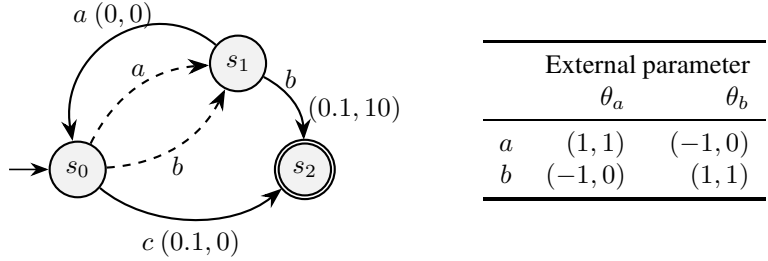

\subsection{Related Work}

Our starting point is the work on 
Bayesian persuasion \cite{kamenica2011bayesian}, which looks at optimal signaling under incomplete information in the \emph{static} case. Many variants of this model have been proposed and studied ever since, with applications in security, voting, advertising, finance, etc. \citep[e.g.,][]{rabinovich2015information,xu2015exploring,goldstein2018stress,badanidiyuru2018targeting,castiglioni2020persuading}; also see the comprehensive surveys \cite{KamenicaSurvey,Dughmi2017}.
Dynamic models of Bayesian persuasion were studied recently \cite{Ely2017,RenaultSolanVieille2017}, and some more recent works focused on algorithmic problems from several dynamic models, such as a model built on extensive-form games (EFGs) \cite{celli2020private} and an online persuasion model \cite{castiglioni2020online,castiglioni2021multireceiver}.
These models are sufficiently different from ours. In the EFG model, in particular, an EFG parameterized by the state of nature (akin to our external parameter) is instantiated before the play, and a group of receivers then engage in the EFG and they infer the EFG being played according to signals from a sender. Hence, information exchange happens only once in this model, whereas it happens in every step in ours. Such one-off persuasions also appeared in several other works on Bayesian persuasion and, more broadly, on non-cooperative IRL (inverse reinforcement learning) and incentive exploration \cite{zhang2019noncooperative,mansour2021bayesian,simchowitz2021exploration}.

%\subsubsection{Dynamic Mechanism Design}
Reversing the roles of the players in terms of who has the power to commit leads to a dual problem of Bayesian persuasion, which is often known as automated mechanism design \cite{conitzer2002complexity,conitzer2004self}.
In such problems, the signal receiver commits to a mechanism that specifies the action they will take upon receiving each signal, and the signal sender sends signals optimally in response. %to this mechanism.
A very recent work by \citet{zhang2021automated} considered automated mechanism design in a dynamic setting  similar to ours, and offered a complementary view to our work.
In their work, the primary consideration is a finite-horizon setting and history-based strategies.
In contrast, we focus primarily on unbounded horizons and memory-less strategies.

%\subsubsection{Stochastic Game}
The interaction between the principal and the agent can be viewed as a stochastic game~\cite{shapley1953stochastic} where one player (i.e., the principal) has the power to make a strategy commitment \cite{letchford2010computing,letchford2012computing}. 
Games where multiple agents jointly take actions in a dynamic environment have been widely studied in the literature on multi-agent reinforcement learning, but usually in settings without strategy commitment \cite{littman1994markov,bucsoniu2010multi}.

More broadly, our work also relates to the advice-based interaction framework \citep[e.g.,][]{torrey2013teaching,DBLP:conf/ijcai/AmirKKG16},  where the principal's goal is to communicate advice to an agent on how to act in the world. 
This advice-based framework is also in close relationship to the machine teaching literature~\cite{goldman1995complexity,DBLP:conf/icml/SinglaBBKK14,doliwa2014recursive,DBLP:journals/corr/abs-1801-05927,ng2000algorithms,hadfield2016cooperative} where the principal (i.e., the teacher) seeks to find an optimal training sequence to steer the agent (i.e., the learner) towards the desired goal. 
Similarly, in environment design, the principal modifies the rewards or transitions to steer the behavior of the agent. The objective may be obtaining fast convergence~\cite{DBLP:conf/icml/NgHR99,DBLP:conf/icml/Mataric94}, or inducing a target policy of the agent~\cite{DBLP:conf/aaai/ZhangP08,DBLP:conf/sigecom/ZhangPC09,DBLP:conf/nips/MaZSZ19,rakhsha2020policy,DBLP:conf/gamesec/HuangZ19a,DBLP:journals/corr/abs-2011-10824}.
These problem settings are similar to ours in that the principal cannot directly act in the environment but can influence the agent's actions via learning signals. We see our setting and techniques as complementary to these studies; in particular, our hardness results can be extended there as well.

\section{The Model}

Our formal model is an MDP with reward uncertainties, given by a tuple $\mathcal{M} = \left\langle S, A, P, \Theta, (\mu_s)_{s \in S}, R, \widetilde{R} \right\rangle$ and involving two players: a {\em principal} and an {\em agent}.
Similar to a standard MDP, $S$ is a finite state space of the environment; 
$A$ is a finite action space {\em for the agent}; 
$P : S \times A \times S \to [0, 1]$ is the transition dynamics of the state.
When the environment is in state $s$ and the agent takes action $a$, the state transitions to $s'$ with probability $P (s, a, s')$; both the principal and the agent are aware of the state throughout.
Meanwhile, rewards are generated for both the principal and the agent, and are specified by the reward functions $R: S \times \Theta \times A \to \mathbb{R}$ and $\widetilde{R}: S \times \Theta \times A \to \mathbb{R}$, respectively. 
Hence, unlike in a standard MDP, here the rewards also depend on an external parameter $\theta \in \Theta$. 
This parameter captures an additional layer of uncertainty of the environment; it follows a distribution $\mu_s \in \Delta(\Theta)$ and is drawn anew every time the state changes.
For all $s \in S$, $\mu_s$ is common prior knowledge shared between the principal and the agent; however, only the principal has access to the realization of $\theta$.

Crucially, since the actions are taken only by the agent, the principal cannot directly influence the state. Instead, the principal can use their information advantage about the external parameter to persuade the agent to take certain actions, by way of signaling.

\subsection{Signaling and Belief Update}

Let $G$ be a space of \emph{signals}. 
A signaling strategy of the principal generates a distribution over $G$.
Our primary consideration in this paper is Markovian signaling strategies, whereby signals to send only depend on the current state (independent of the history).
Formally, a signaling strategy $\pi = (\pi_s)_{s \in S}$ of the principal consists of a function $\pi_s: \Theta \to \Delta(G)$ for each state $s \in S$.
Upon observing an external parameter $\theta$, the principal will send a signal sampled from $\pi_s(\theta)$ when the current state is $s$; we denote by $\pi_s(\theta, g)$ the probability of $g\in G$ in this distribution.

The signal space is broadly construed.
For example, one simple signaling strategy is to always reveal the true information, which always sends a deterministic signal $g_\theta$ associated with the observed external parameter $\theta \in \Theta$ (i.e., a message saying ``The current external state is $\theta$''); formally, we write $\pi_s(\theta) = \mathbf{\hat {e}}_{g_\theta}$.\footnote{We let $\mathbf{\hat {e}}_i$ denote a unit vector, of which the $i$-th element is $1$.}
In contrast, if the same signal is sent irrespective of the external parameter, i.e., $\pi_s(\theta) = \pi_s(\theta')$ for all $\theta, \theta' \in \Theta$, then the signaling strategy is completely uninformative. 
Without loss of generality, we assume that signals in $G$ are distinct from each other from the agent's point of view.

Upon receiving a signal $g$, the agent updates their posterior belief about the (distribution of) the external parameter: the conditional probability of the parameter being $\theta$ is
\begin{equation}
\label{eq:posterior}
\textstyle
\Pr(\theta|g,\pi_s) = \frac{\mu_s(\theta) \cdot \pi_s(\theta, g)}{\sum_{\theta'\in \Theta} \mu_s(\theta') \cdot \pi_s(\theta', g)}.
\end{equation}
To derive the above posterior also relies on knowledge about the principal's signaling strategy $\pi$.
Indeed, we follow the Bayesian persuasion framework, whereby the principal commits to a signaling strategy $\pi$ at the beginning of the game and announces it to the agent.

\subsection{Signaling Strategy Optimization}

We take the principal's point of view and investigate the problem of optimal signaling strategy design:
given $\mathcal{M}$, find a signaling strategy $\pi$ that maximizes the principal's (discounted) cumulative reward 
$\mathbb{E} \left[ \sum_{t=0}^{\infty} \gamma^t R(s_t, \theta_t, a_t) | \mathbf{z}, \pi \right]$,
where $\mathbf{z} = (z_s)_{s \in S}$ is the distribution of the starting state, 
$\gamma \in [0,1)$ is a discount factor,
and the expectation is taken over the trajectory $(s_t, \theta_t, a_t)_{t=0}^\infty$ induced by the signaling strategy $\pi$.
To completely specify this task requires a behavioral model for the agent.
We will consider two major types of agents---{\em myopic} and {\em far-sighted}---and will define them separately in the next two sections; 
a myopic agent only cares about their instant reward in each step, whereas a far-sighted agent considers the cumulative reward with respect to a discount factor $\tilde{\gamma} > 0$ (which need not be equal to $\gamma$).

In summary, the game proceeds as follows.
At the beginning, the principal commits to a signaling strategy $\pi$ and announces it to the agent.
Then in each step, if the environment is in state $s$, an external parameter $\theta \sim \mu_s$ is drawn (by nature); the principal observes $\theta$, samples a signal $g \sim \pi_s(\theta)$, and sends $g$ to the agent.
The agent receives $g$, updates their belief about $\theta$ (according to \eqref{eq:posterior}), and decides an action $a \in A$ to take accordingly.
The state then transitions to $s' \sim P(s, a, \cdot)$.

\section{When Agent is Myopic}
\label{sc:opt-myopic}

We first consider a myopic agent.
A myopic agent aims to maximize their reward in each individual step.
Upon receiving a signal $g$ in state $s$, the agent will take a best action $a \in A$, which maximizes
$\mathbb{E}_{\theta \sim \Pr(\cdot|g,\pi_s)} \widetilde{R}(s,\theta, a)$.
We study the problem of computing an optimal signaling strategy against a myopic agent, termed {\sc OptSig-myop}.

\subsection{Action Advice}

According to a standard argument via the revelation principle, it is often without loss of generality to consider signaling strategies in the form of action advice.
This also holds in our model.
Specifically, for any signaling strategy, there exists an equivalent strategy $\pi$ which uses only a finite set $G_A := \{g_a : a \in A \}$ of signals, and each signal $g_a$ corresponds to an action $a \in A$;
moreover, $\pi$ is {\em incentive compatible} (IC), which means that the agent is also incentivized to take the corresponding action $a$ upon receiving $g_a$, i.e., we have $\mathbb{E}_{\theta \sim \Pr(\cdot|g_a,\pi_s)} \widetilde{R}(s,\theta, a) \ge \mathbb{E}_{\theta \sim \Pr(\cdot|g_a,\pi_s)} \widetilde{R}(s,\theta, a')$ for all $a'\in A$,\footnote{By convention, we assume that the agent breaks ties by taking the advised action when there are multiple optimal actions.} or equivalently:
\begin{align}
\label{eq:ic}
\sum_{\theta \in \Theta} \Pr(\theta|g_a,\pi_s) {\cdot} \left( \widetilde{R}(s,\theta, a) - \widetilde{R}(s,\theta, a') \right) \ge 0 
\quad \text{for all } a' \in A.
\end{align}
In other words, $\pi$ signals which action the agent should take and it is designed in a way such that the agent cannot be better off deviating from the advised action with respect to the posterior belief.
We call a signaling strategy that only uses signals in $G_A$ an {\em action advice}, and call it an {\em IC} action advice if it also satisfies \eqref{eq:ic}.
We refer the reader to the appendix for more details about the generality of IC action advices in our model.

We can easily characterize the outcome of an IC action advice $\pi$:
at each state $s$, since the agent is incentivized to follow the advice, with probability $\phi_s^\pi(\theta, a) := \mu_s(\theta) \cdot \pi_s(\theta, g_a)$ they will take action $a$ while the realized external parameter is $\theta$; hence, $\phi_s^\pi$ is a distribution over $\Theta \times A$.
We can then define the following set $\mathcal{A}_s  \subseteq \Delta(\Theta \times A)$, which contains all such distributions that can be induced by some $\pi$:
\[
\mathcal{A}_s  = \left\{ \phi_s^\pi : \pi \text{ is an IC action advice}\right\}.
\]
It would now be convenient to view the problem facing the principal as an (single-agent) MDP $\mathcal{M}^* = \left\langle S, (\mathcal{A}_s)_{s \in S}, P^*, R^* \right\rangle$, where
$S$ is the same state space in $\mathcal{M}$; 
$\mathcal{A}_s$ defines an (possibly infinite) action space for each $s$; 
the transition dynamics $P^*: S \times \Delta(\Theta \times A) \times S \to [0,1]$ and reward function $R^*: S \times \Delta(\Theta \times A) \to \mathbb{R}$ are such that 
\begin{align*}
P^*(s, \mathbf{x}, s') &= \mathbb{E}_{(\theta, a) \sim \mathbf{x}} P(s, a , s'),\\ 
\text{ and } \quad 
R^*(s, \mathbf{x}) \quad &= \mathbb{E}_{(\theta, a) \sim \mathbf{x}} R(s, \theta, a).
\end{align*}
Namely, $\mathcal{M}^*$ is defined as if the principal can choose actions (which are $(\theta, a)$ pairs) freely from $\mathcal{A}_s$, whereas the choice is actually realized through persuasion.
A policy $\sigma$ for $\mathcal{M}^*$ maps each state $s$ to an action $\mathbf{x} \in \mathcal{A}_s$, and it corresponds to an IC action advice $\pi$ in $\mathcal{M}$, with $\phi_s^\pi = \sigma(s)$ for all $s$.
The problem of designing an optimal action advice then translates to computing an optimal policy for $\mathcal{M}^*$.
We show next that we can exploit a standard approach to compute an optimal policy but we need to address a key challenge as the action space of $\mathcal{M}^*$ may contain infinitely many actions.

\subsection{LP Formulation}

The standard approach to computing an optimal policy for an MDP is to compute a value function $V:S \to \mathbb{R}$ that satisfies the Bellman equation:
\begin{align}
\label{eq:bellman}
V(s) = \max_{\mathbf{x} \in \mathcal{A}_s} \left[ R^*(s, \mathbf{x}) + \gamma \cdot \sum_{s' \in S} P^*(s, \mathbf{x}, s') \cdot V(s') \right] \quad \text{for all } s \in S.
\end{align}
It is well-known that there exists a unique solution to the above system of equations, from which an optimal policy can be extracted.
In particular, one approach to computing this unique solution is by using the following LP (linear program) formulation, where $V(s)$ are the variables;
The optimal value of this LP directly gives the cumulative reward of optimal policies under a given initial state distribution $\mathbf{z}$. 
\begin{align}
\min_{V}
\quad &  \sum_{s \in S} z_s \cdot V(s) \label{lp:opt-myop}
\\
\text{s.t.} \quad
&V(s) \ge R^*(s, \mathbf{x}) + \gamma \cdot \sum_{s' \in S} P^*(s, \mathbf{x}, s') \cdot V(s') && \text{for all } s \in S, \mathbf{x} \in \mathcal{A}_s \tag{\ref{lp:opt-myop}a}
\label{eq:opt-myop-cons-1}
\end{align}
The issue with this LP formulation is that there may be infinitely many constraints as \eqref{eq:opt-myop-cons-1} must hold for all $\mathbf{x} \in \mathcal{A}_s$.
This differs from MDPs with a finite action space, in which case the LP formulation can be reduced to one with a finite set of constraints, where each constraint corresponds to an action. 
We address this issue by using the {\em ellipsoid method} as sketched below. More practically, we can also derive a concise LP formulation by exploiting the duality principle (see the appendix). 

\begin{theorem}
{\sc OptSig-myop} is solvable in polynomial time.
\end{theorem}

% ===============================
% -- Proof by Ellipsoid Method --
% ===============================
\begin{proof}[Proof sketch]
We show that LP~\eqref{lp:opt-myop} can be solved in polynomial time by using the ellipsoid method.
The key to this approach is to implement the {\em separation oracle} in polynomial time. For any given value assignment of the variables (in our problem values of $V(s)$), the oracle should decide correctly whether all the constraints of the LP are satisfied and, if not, output a violated one.

To implement the separation oracle for our problem amounts to solving the following optimization for all $s \in S$:
\[
\max_{\mathbf{x} \in \mathcal{A}_s}
\quad R^*(s, \mathbf{x}) + \gamma \cdot \sum_{s' \in S} P^*(s, \mathbf{x}, s') \cdot V(s') - V(s).
\]
By checking if the above maximum value is positive, we can identify if \eqref{eq:opt-myop-cons-1} is violated for some $\mathbf{x} \in \mathcal{A}_s$.
Indeed, the set of IC action advices can be characterized by the constraints in \eqref{eq:ic}, which are linear constraints if we expand $\Pr(\theta|g_a, \pi_s)$ according to \eqref{eq:posterior} and eliminate the denominator (where we also treat $\pi_s(\theta, g_a)$ as the additional variables and add the constraint $x(\theta, a) = \mu_s(\theta) \cdot \pi_s(\theta, g_a)$ for every $\theta$ and $a$).
\end{proof}

\section{When Agent is Far-sighted}
\label{sec:farsighted}

A far-sighted (FS) agent looks beyond the immediate reward and considers the cumulative reward discounted by $\tilde{\gamma}$. We now study signaling strategy design against an FS agent.

\subsection{Optimal Signaling against FS Agent}
\label{sc:opt-fs}

When facing an FS agent, we cannot define an inducible set $\mathcal{A}_s$ independently for each state. 
The principal needs to take a global view and aim to induce the agent to use a {\em policy} that benefits the principal.
We term the problem of optimal signaling strategy design against an FS agent {\sc OptSig-FS}.

\subsubsection{Best Response of FS Agent}

% We start with the problem from the agent's perspective.
We first investigate an FS agent's best response problem.
When the principal commits to a signaling strategy $\pi$, the best response problem facing the agent can be formulated as an MDP $\mathcal{M}^\pi = \left\langle S\times G, A, P^\pi, \widetilde{R}^\pi \right\rangle$. In each step the agent observes the state $s\in S$ of $\mathcal{M}$ along with a signal $g \in G$ from the principal; the tuple $(s,g)$ constitutes a state in $\mathcal{M}^\pi$, and we call it a {\em meta-state} to distinguish it from states in $\mathcal{M}$.
From the agent's perspective, after they take action $a$, the meta-state transitions to $(s', g')$ with probability 
\begin{equation}
\label{eq:P-star}
P^\pi((s,g), a, (s', g')) = P(s, a, s') \cdot \mathbb{E}_{\theta \sim \mu_{s'}} \pi_{s'}(\theta, g').
% P^\pi((s,g), a, (s', g')) = P(s, a, s') \cdot \sum_{\theta' \in \Theta} \mu_{s'}(\theta') \cdot \pi_{s'}(\theta', g').
\end{equation}
Namely, a next state $s'$ of $\mathcal{M}$ is sampled from $P(s,a,\cdot)$, then a new external parameter $\theta$ is sampled from $\mu_{s'}$ and the principal sends a signal $g \sim \pi_{s'}(\theta)$.
Meanwhile, the following reward is yielded for the agent:
\begin{align}
\widetilde{R}^\pi((s,g), a) 
&= \mathbb{E}_{\theta \sim \Pr(\cdot|g, \pi_s)} \widetilde{R}(s,\theta, a), 
\label{eq:tR-star}
\end{align}
where the posterior belief $\Pr(\theta|g, \pi_s)$ is defined in \eqref{eq:posterior}.

Hence, an optimal policy ${\sigma}:S\times G \to A$ for $\mathcal{M}^\pi$ defines a best response of the agent against $\pi$. 
An optimal signaling strategy of the principal maximizes the cumulative reward against the agent's best response.

\paragraph{Inapproximability}
We show that {\sc OptSig-FS} is highly intractable:
even to find an approximate solution to {\sc OptSig-FS} requires solving an NP-hard problem. 
Hence, it is unlikely that there exists any efficient approximation algorithm for this task, assuming that P=NP is unlikely.

\begin{theorem}
\label{thm:harness-optfs}
Assuming that P $\neq$ NP, then {\sc OptSig-FS} does not admit any polynomial-time $\frac{1}{\lambda^{1-\epsilon}}$-approximation algorithm for any constant $\epsilon > 0$, where $\lambda$ is the number of states $s \in S$ in which the prior distribution $\mu_s$ is non-deterministic (i.e., supported on at least two external parameters).
This holds even when $|\Theta|=2$ and the discount factors $\gamma, \tilde{\gamma} \in (0,1)$ are fixed.
\end{theorem}

The proof of Theorem~\ref{thm:harness-optfs} is via a reduction from the {\sc Maximum Independent Set} problem,
%\footnote{All the omitted proofs can be found in the full version of this paper.} 
which is known to be NP-hard to approximate \cite{zuckerman2006linear}. The result may also be of independent interest: It can be easily adapted to show the inapproximability of similar principal-agent problems in dynamic settings. 
%(see remarks in Appendix~\ref{sc:proof-hardness}).
This hardness result also indicates a ``phase transition'' between the cases where $\tilde{\gamma} = 0$ and $\tilde{\gamma} > 0$ given the tractability of {\sc OptSig-myop} showed in Section~\ref{sc:opt-myopic}.

\subsection{Advice-myopic Agent}
\label{sc:adv-myop}

The intractability of {\sc OptSig-FS} motivates us to consider \emph{advice-myopic (AM)} agents, who account for their future rewards like an FS agent does, but behave myopically and ignore the principal's future signals.
In other words, they always assume that the principal will disappear in the next step and rely only on their prior knowledge to estimate the future payoff.
We refer to the optimal signaling strategy problem against an AM agent as {\sc OptSig-AM}.

\paragraph{Equivalence to the Myopic Setting}

Since an AM agent does not consider future signals, their future reward is independent of the principal's signaling strategy. 
This allows us to define a set of inducible distributions of $(\theta, a)$ independently for each state, similarly to our approach to dealing with a myopic agent.
In other words, an AM agent is equivalent to a myopic agent who adds a fixed value to their reward function, and this fixed value is the best future reward they can achieve without the help of any signals. This value is independent of the signaling strategy and can be calculated beforehand.
Let $\widetilde{R}^+: S \times \Theta \times A \to \mathbb{R}$ be the reward function of this equivalent myopic agent.
We have 
\begin{equation}
\label{eq:tR-p}
\widetilde{R}^+(s,\theta,a) = \widetilde{R}(s, \theta, a) + \tilde{\gamma} \cdot \mathbb{E}_{s' \sim P(s,a,\cdot)} \overline{V}(s',g_0),
\end{equation}
where $\overline{V}$ is the optimal value function of the agent when completely uninformative signals are given.
In more detail, let $\pinothing: \Theta \to \Delta(G)$ be a completely uninformative signaling strategy, with $\pinothing(\theta) = \mathbf{\hat{e}}_{g_0}$ for all $\theta$ (i.e., it always sends the same signal $g_0$).
Then $\overline{V}$ is the optimal value function for the MDP $\mathcal{M}^\pinothing =  \left\langle S\times \{g_0\}, A, P^\pinothing, \widetilde{R}^\pinothing \right\rangle$, defined the same way as $\mathcal{M}^\pi$ in Section~\ref{sc:opt-fs}, with $\pi = \pinothing$.
%
%% !!Line break removed to save space
Hence, the Bellman equation gives
% \iftoggle{shortversion}{
\begin{align}
&\overline{V}(s, g_0) \nonumber\\
=& \max_{a \in A} \left( \widetilde{R}^\pinothing (s,\theta,a) + \tilde{\gamma} \cdot \mathbb{E}_{(s',g_0) \sim P^\pinothing((s,g_0), a, \cdot)} \overline{V}(s',g_0) \right) \hspace{-20mm} \nonumber\\
=& \max_{a \in A} \left( \mathbb{E}_{\theta \sim \mu_s}  \widetilde{R}(s,\theta,a) + \tilde{\gamma} \cdot \mathbb{E}_{s' \sim P(s,a, \cdot)} \overline{V}(s',g_0) \right) 
\label{eq:bellman-tV-g0}
\end{align}
for all $s \in S$,
where the second transition follows by \eqref{eq:P-star} and \eqref{eq:tR-star} and we also use the facts that the posterior $\Pr(\cdot|g,\pinothing)$ degenerates to the prior $\mu_s(\cdot)$ as $\pinothing$ is uninformative, and that $P^\pinothing((s,g_0), a, (s',g_0)) = P(s,a,s')$ as the meta-state only transitions among the ones in the form $(s, g_0)$.

We can compute $\overline{V}$ efficiently by solving the above Bellman equation. (A standard LP approach suffices given that $\mathcal{M}^\pinothing$ has a finite action space.) Then we obtain $\widetilde{R}^+$ according to \eqref{eq:tR-p}, with which we can construct an equivalent {\sc OptSig-myop} instance and solve it using our algorithm in Section~\ref{sc:opt-myopic}.
The solution is also optimal to the original {\sc OptSig-AM} instance as we argued above; we state this result in the theorem below and omit the proof.

\begin{theorem}
{\sc OptSig-AM} is solvable in polynomial time.
\end{theorem}

\subsection{Threat-based Action Advice against FS Agent}

Now that we can efficiently solve {\sc OptSig-AM}, we will show that we can use a solution to {\sc OptSig-AM} to efficiently design a signaling strategy against an FS agent.
Interestingly, we can prove that this strategy guarantees the principal the payoff as if they are playing against an AM agent, when the agent is actually FS. 
The idea is to add a threat in the action advice: if the agent does not take the advised action, then the principal will stop providing any information in future steps (equivalently, switching to strategy $\pinothing$).
Essentially, this amounts to a one-memory strategy, denoted $\varpi = \left(\varpi_s \right)_{s \in S}$, where each $\varpi_s:S \times \Theta \times G \times A \to \Delta(A)$ also depends on the signal and the action taken in the previous step (i.e., whether the action follows the signal).

More formally, suppose that $\pi = (\pi_s)_{s \in S}$ is a solution to {\sc OptSig-AM} and without loss of generality it is an IC action advice.
We construct a one-memory strategy:
% \iftoggle{shortversion}{
\begin{equation}
\label{eq:pi-mmone}
\varpi_s((s,\theta), g, a) = 
\begin{cases}
\pi_s (\theta), & \text{if } g \in \{ g_a , \textit{null}\} \\
\pinothing(\theta) = \mathbf{\hat{e}}_{g_0}, & \text{otherwise}
\end{cases}
\end{equation}
where $g$ and $a$ are the signal and action taken in the previous step (assume that $g$ is initialized to \textit{null} in the first step); each signal $g_a$ advises the agent to take the corresponding action $a$, and $g_0$ is a signal that does not correspond to any action.

Our key finding is that, via this simple threat-based mechanism, the strategy $\varpi$ we design is persuasive for an FS agent: the threat it makes effectively incentivizes the FS agent to take advised actions. 
To show this, we first analyze the problem facing the agent when the principal commits to $\varpi$.

\paragraph{Best Response to $\varpi$}

From an FS agent's perspective, the principal committing to $\varpi$ results in an MDP $\mathcal{M}^\varpi =\left\langle S\times G, A, P^\varpi, \widetilde{R}^\varpi \right\rangle$.
We have $G = \{g_0\} \cup G_A$, so each meta-state $(s, g)$ in $\mathcal{M}^\varpi$ consists of a state of $\mathcal{M}$ and a signal from the principal. 
The transition dynamics depend on whether the signal sent in the current state is $g_0$ or not (i.e., whether the principal has switched to the threat-mode):

\begin{itemize}
\item
For all $(s, g_a) \in S\times G_A$, the agent following the advised action $a$ results in transition probabilities:
\begin{subequations}
\label{eq:thm-varpi-P-varpi}
\begin{align}
& P^\varpi((s,g_a), a, \cdot) =  P^\pi((s,g_a), a, \cdot); \label{eq:thm-varpi-P-varpi-a} \end{align}
\end{subequations}
Otherwise, i.e., if any action $b \neq a$ is taken, the principal will fulfill the threat and send $g_0$ in the next step. Hence, 
%\iftoggle{shortversion}{
\begin{align}
P^\varpi((s,g_a), b, (s',g)) =
\begin{cases}
\sum_{g'\in G} P^\pi((s,g_a), a, (s',g')), & \text{if } g = g_0\\
0, & \text{otherwise}
\end{cases}
\tag{\ref{eq:thm-varpi-P-varpi}b} \label{eq:thm-varpi-P-varpi-b}
\end{align}

\item
For all $(s, g_0) \in S\times \{g_0\}$, the threat is activated in these meta-states, we have:
% \iftoggle{shortversion}{
\begin{align}
P^{\varpi}((s,g_0), a, (s', g')) 
= P^{\pinothing}((s,g_0), a, (s', g')) = 
\begin{cases}
P(s,a, s'), & \text{if } g' = g_0 \\
0, & \text{otherwise}
\end{cases}
\tag{\ref{eq:thm-varpi-P-varpi}c} 
\label{eq:thm-varpi-P-varpi-c}
\end{align}

\end{itemize}

Similarly, the reward function differs in meta-states $(s,g_a)$ and $(s,g_0)$. We have
%\begin{subequations}
\begin{align}
\widetilde{R}^\varpi((s,g), \cdot) =  
\begin{cases}
\widetilde{R}^\pi((s,g), \cdot), &\text{ if } g \in G_A \\
\widetilde{R}^{\pinothing}((s,g_0), \cdot), & \text{ otherwise}
\end{cases}
\label{eq:thm-varpi-tR-varpi}
\end{align}
%\end{subequations}

\paragraph{Persuasiveness of $\varpi$}

To show the persuasiveness of $\varpi$, we argue that the following policy $\sigma: S\times G \to A$ of the agent, in which the agent always takes the advised action, is optimal in response to $\varpi$.
For all $s \in S$, we define
\begin{equation}
\label{eq:opt-policy-to-varpi}
\sigma(s,g) =  
\begin{cases}
a, & \text{if } g = g_a \in G_A \\
\bar{\sigma}(s,g_0), &  \text{if } g = g_0
\end{cases}
\end{equation}
where $\bar{\sigma}$ is an optimal policy against $\pinothing$, the value function $\overline{V}: S \times G \to \mathbb{R}$ of which (as from the agent's perspective) is already defined in Section~\ref{sc:adv-myop} and satisfies \eqref{eq:bellman-tV-g0}.

Theorem~\ref{thm:varpi-IC} shows the optimality of $\sigma$. Intuitively, we show that the value function of $\sigma$ is at least as large as $\overline{V}$, so the agent has no incentive to provoke the threat.

\begin{theorem}
\label{thm:varpi-IC}
The policy $\sigma$ defined in \eqref{eq:opt-policy-to-varpi} is an optimal response of an FS agent to $\varpi$. (Hence, $\varpi$ incentivizes an FS agent to take the advised action.)
\end{theorem}

The direct consequence of Theorem~\ref{thm:varpi-IC} is that $\varpi$ guarantees the principal the best payoff they can obtain when facing an AM agent, even though the agent is FS (Corollary~\ref{crl:varpi-IC}).
Hence, $\varpi$ serves as an alternative approach to deal with an FS agent.
Note that this threat-based strategy may not be an optimal one-memory strategy.
Indeed, with minor changes to our proof of Theorem~\ref{thm:harness-optfs}, we can show that for any positive integer $k$ the problem of computing an optimal $k$-memory strategy is inapproximable (see the appendix).
In contrast, in the myopic and advice-myopic settings, since the agent's behavior is Markovian, the optimal signaling strategies we designed remain optimal even when we are allowed to use memory-based strategies.

\begin{corollary}
\label{crl:varpi-IC}
By using $\varpi$ against an FS agent, the principal's cumulative reward is the same as the highest cumulative reward they can obtain against an AM agent.
\end{corollary}

% ================= EXPERIMENTS ================= %

% ==================== %
% --    RESULTS.    -- %
% ==================== %

% gamma_pr = gamma_ag = 0.8, n_termin = 0, S = 10, Theta = 10, A = 10

% gamma_pr = gamma_ag = 0.8, n_termin = 5, S = 10, Theta = 10, A = 10

% gamma_pr = gamma_ag = 0.8, n_termin = 5, S = 10, Theta = 10, A = 2, 4, ..., 20

% gamma_pr = gamma_ag = 0.8, n_termin = 5, S = 10, A = 10, Theta = 2, 4, ..., 20

% gamma_pr = gamma_ag = 0.8, n_termin = 0, Theta = 10, A = 10, S = 2, 4, ..., 20

% gamma_pr = gamma_ag = 0.8, n_termin = 0, Theta = 10, A = 10, S = 2, 4, ..., 20

% gamma_pr=0.1,0.2,...,1.0, gamma_ag=0.8, n_termin = 5, Theta = 10, A = 10, S = 10

% gamma_pr = 0.8, gamma_ag=0.1,0.2,...,1.0, n_termin = 5, Theta = 10, A = 10, S = 10

% ====================== %
% --  FIGURE SETTING  -- %
% ====================== %

%\newlength\figurewidth
%\setlength\figurewidth{.27\linewidth}
%\newlength\figureheight
%\setlength\figureheight{.26\linewidth}

\definecolor{myRed}{RGB}{230,20,0}
\definecolor{myGreen}{RGB}{20,190,0}
\definecolor{myBlue}{RGB}{0,20,210}

\newenvironment{myplot}[1]
{
	\begin{tikzpicture}
	\begin{axis}[
	width=\figurewidth,
	height=\figureheight,
	xlabel={\empty},
	ylabel={\empty},
	xmin=-1, xmax=1,
	ymin=0, ymax=1.0,
    xtick={ -1, -.5, 0, .5, 1},
	minor x tick num=1,
    minor y tick num=0,
    minor tick length=0.4mm,
    xticklabels={$-1$,$-.5$,$0$,$.5$,$1$},
	ytick={0.2,0.4,0.6,0.8,1.0},
    yticklabels={$.2$,$.4$,$.6$,$.8$,$1$},
    tick label style = {font={\footnotesize}},
	major tick length=0.6mm,
	every tick/.style={
		black,
	},
	legend pos=north east,
	xmajorgrids=true,
    ymajorgrids=true,
    %yminorgrids=true,
	grid style=dotted,
	%style={font=\small},
	%label style={font={\footnotesize}}, %\sansmath\sffamily}},
	%legend style={legend columns=1, row sep=-0.3mm, inner sep=2pt, font={\scriptsize}},
	legend style={font=\footnotesize, /tikz/every even column/.append style={column sep=3mm}, legend columns=7, inner sep=0.5mm, anchor=north east, },
	legend cell align=left,
	x label style={font={\tiny},at={(axis description cs:0.5,0.0)}},
	y label style={font={\small},at={(axis description cs:0.15,0.48)}},
    %x tick label style = {font=\small},
    %y tick label style = {font=\small},
	cycle list={%
		{myBlue, semithick},    % noSigMyop
		{myBlue, densely dotted, thick}, % noSigFS
		{myRed, semithick}, % optSigMyop
		{myRed, densely dotted, thick}, % optSigAM
		%{myRed, semithick, mark=o, mark options={solid}, mark size=1.0pt}, 
		%{myGreen, semithick, mark=x, mark options={solid}, mark size=1.7pt}, 
	},
	title style={font={\small}, at={(axis description cs:0.5,-0.5)}},
	#1,
	]
}{
\end{axis}
\end{tikzpicture}
}

% ----------------- PLOT STANDARD ERROR BAND -----------------
\newcommand{\errorband}[5][]{ % x column, y column, error column, optional argument for setting style of the area plot
\pgfplotstableread[col sep=space]{#2}\datatable
    \addplot [draw=none, stack plots=y, forget plot] table [
        x={#3},
        y expr=\thisrow{#4}-\thisrow{#5}
    ] {\datatable};

    % Stack twice the error, draw as area plot
    \addplot [draw=none, fill=gray!40, stack plots=y, area legend, #1] table [
        x={#3},
        y expr=2*\thisrow{#5}
    ] {\datatable} \closedcycle;

    % Reset stack using invisible plot
    \addplot [forget plot, stack plots=y,draw=none] table [x={#3}, y expr=-(\thisrow{#4}+\thisrow{#5})] {\datatable};
}

% ----------------- ADD PLOT COMMANDS ----------------- 
\newcommand{\addplotmyop}[1]{%
    \addplot table[x=x, y=noSigMyop] {#1};
    \addplot coordinates{(0,0)};
    \addplot table[x=x, y=optSigMyop] {#1};
    \addplot coordinates{(0,0)};
    \errorband[myBlue, opacity=0.2]{#1}{x}{noSigMyop}{StdDev_noSigMyop}
    \errorband[myRed, opacity=0.2]{#1}{x}{optSigMyop}{StdDev_optSigMyop}
    }
    
\newcommand{\addplotfs}[1]{%
    \addplot coordinates{(0,0)};
    \addplot table[x=x, y=noSigFS] {#1};
    \addplot coordinates{(0,0)};
    \addplot table[x=x, y=optSigAM] {#1};
    \errorband[myBlue, opacity=0.2]{#1}{x}{noSigFS}{StdDev_noSigFS}
    \errorband[myRed, opacity=0.2]{#1}{x}{optSigAM}{StdDev_optSigAM}
    }
% -------------------------------------------------------------

\begin{NoHyper}
\begin{figure*}[h]
\centering
%\ref{leg:all}
\begin{tabular}{@{\hspace{0mm}}l@{\hspace{0mm}}l@{\hspace{0mm}}l}
% ========================
\begin{myplot}
    {
        title={(a)\quad $\beta$},
        legend entries={ {\sc NoSig-myop}, {\sc NoSig-AM/FS}, {\sc OptSig-myop}, {\sc OptSig-AM (Threat-FS)} },%
 	    legend to name=leg:all,
    }
    \addplotmyop{random_termin_5.dat}
\end{myplot}%\\[.5mm]
&
% ========================
\begin{myplot}
    {
        title={(b)\quad $\beta$},
        yticklabels={\empty},
    }
    \addplotfs{random_termin_5.dat}
\end{myplot}%
&
% ========================
\begin{myplot}
    {
        title={(c)\quad $|A|$},
        yticklabels={\empty},
        xmin=2, xmax=20,
        xtick={ 4, 8, 12, 16, 20},
        xticklabels={$4$,$8$,$12$,$16$,$20$},
    }
    \addplotmyop{random_var_A.dat}
\end{myplot}%
\\[2mm]
% ========================
\begin{myplot}
    {
        title={(d)\quad $|\Theta|$},
        xmin=2, xmax=20,
        xtick={ 4, 8, 12, 16, 20},
        xticklabels={$4$,$8$,$12$,$16$,$20$},
    }
    \addplotmyop{random_var_Theta.dat}
\end{myplot}%\\[1mm]
&
% ========================
\begin{myplot}
    {
        title={(e)\quad $|S|$ ($n^*=0$)},
        yticklabels={\empty},
        xmin=2, xmax=20,
        xtick={ 4, 8, 12, 16, 20},
        xticklabels={$4$,$8$,$12$,$16$,$20$},
    }
    \addplotmyop{random_var_S.dat}
\end{myplot}%
&
% ========================
\begin{myplot}
    {
        title={(f)\quad $n^*$},
        yticklabels={\empty},
        xmin=0, xmax=9,
        xtick={0, 2, 4, 6, 8},
        xticklabels={$0$,$2$,$4$,$6$,$8$},
    }
    \addplotmyop{random_var_n_termin.dat}
\end{myplot}%
\end{tabular}
%\vspace{-3mm}
\ref{leg:all}
\caption{Comparison of signaling strategies: all results are shown as ratios to {\sc FullControl} on the y-axes. Meanings of x-axes are noted in the captions. Shaded areas represent standard deviations (mean $\pm$ standard deviation). 
In all figures, we fix $|S| = |\Theta| = |A| = 10$, $\gamma = \tilde{\gamma} = 0.8$, $n^*=5$, and $\beta=0$ unless they are variables.}
\label{fig:exp}
\end{figure*}
\end{NoHyper}

% ================== ROAD NAVIGATION ========================

% gamma_pr = gamma_ag = 0.8, node = 20, edge = 100, theta = 3, beta = -1,...,1

% gamma_pr = gamma_ag = 0.8, node = 20, edge = 100, theta = 2,...,10, beta = 0.5

% gamma_pr = gamma_ag = 0.8, node = 20, edge = 20,...100, theta = 3, beta = 0.5

\begin{NoHyper}
\begin{figure*}
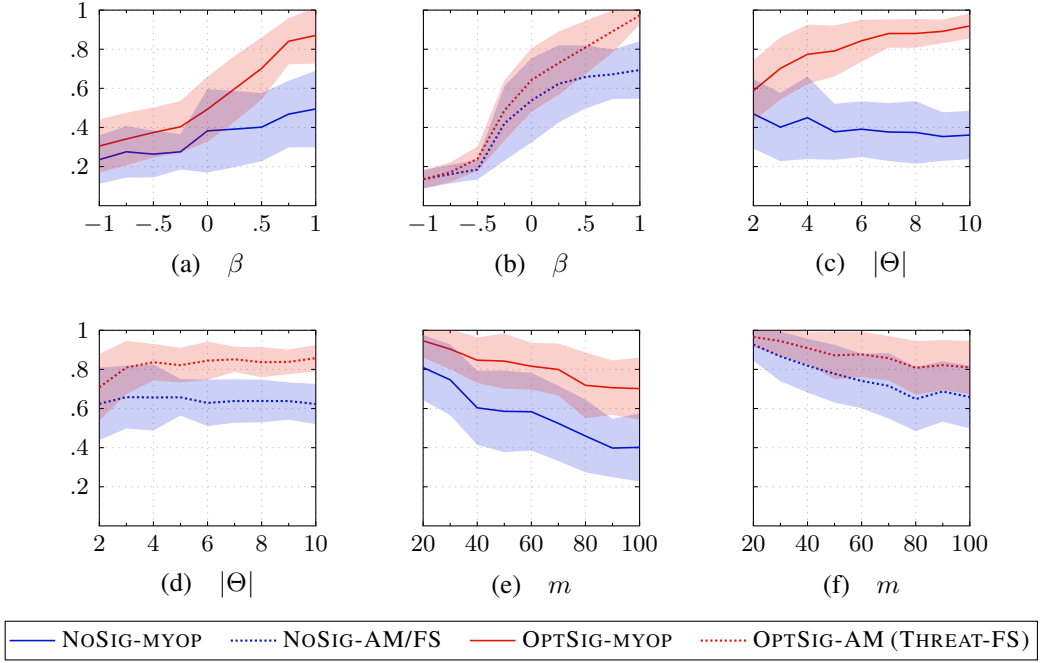

%\vspace{-2mm}
\centering
\begin{tabular}{@{\hspace{0mm}}l@{\hspace{0mm}}l@{\hspace{0mm}}l}
% ========================
\begin{myplot}
    {
        title={(a)\quad $\beta$},
    }
    \addplotmyop{road_nav_beta.dat}
\end{myplot}%\\[.5mm]
&
\begin{myplot}
    {
        title={(b)\quad $\beta$},
        yticklabels={\empty},
    }
    \addplotfs{road_nav_beta.dat}
\end{myplot}
&
\begin{myplot}
    {
        title={(c)\quad $|\Theta|$},
        yticklabels={\empty},
        xmin=2, xmax=10,
        xtick={ 2, 4, 6, 8, 10},
        xticklabels={$2$, $4$, $6$, $8$, $10$},
    }
    \addplotmyop{road_nav_Theta.dat}
\end{myplot}
\\[2mm]
\begin{myplot}
    {
        title={(d)\quad $|\Theta|$},
        xmin=2, xmax=10,
        xtick={ 2, 4, 6, 8, 10},
        xticklabels={$2$, $4$, $6$, $8$, $10$},
    }
    \addplotfs{road_nav_Theta.dat}
\end{myplot}
&
\begin{myplot}
    {
        title={(e)\quad $m$},
        yticklabels={\empty},
        xmin=20, xmax=100,
        xtick={ 20, 40, 60, 80, 100},
        xticklabels={$20$, $40$, $60$, $80$, $100$},
    }
    \addplotmyop{road_nav_road.dat}
\end{myplot}
&
\begin{myplot}
    {
        title={(f)\quad $m$},
        yticklabels={\empty},
        xmin=20, xmax=100,
        xtick={ 20, 40, 60, 80, 100},
        xticklabels={$20$, $40$, $60$, $80$, $100$},
    }
    \addplotfs{road_nav_road.dat}
\end{myplot}
% ========================
\end{tabular}
\ref{leg:all}
%\vspace{-4mm}
\caption{Comparison of signaling strategies in a navigation application: 
all results are shown as the ratios of {\sc FullControl} to them on the y-axes (now that rewards are costs).
All curves and axes have the same meanings as in Figure~\ref{fig:exp}. 
All results are obtained on instances with $n=20$ and $m=100$ (i.e., numbers of nodes and edges in the network), where we also fix $|\Theta| =3$, $\gamma = \tilde{\gamma} = 0.8$, and $\beta=0.5$ unless they are variables.
}
%\vspace{-2mm}
\label{fig:exp_nav}
\end{figure*}
\end{NoHyper}

% gamma_pr = gamma_ag = 0.8, node = 20, edge = 100, theta = 3, beta = -1,...,1

% gamma_pr = gamma_ag = 0.8, node = 20, edge = 100, theta = 2,...,10, beta = 0.5

% gamma_pr = gamma_ag = 0.8, node = 20, edge = 20,...100, theta = 3, beta = 0.5

\section{Experiments}
\label{sec:experiments}

We empirically evaluate signaling strategies obtained with our algorithms. 
The goal is to compare the payoffs yielded for the principal.
We use Python (v3.9) to implement our algorithms and Gurobi (v9.1.2) to solve all the LPs. All results were obtained on a platform with a 2 GHz Quad-Core CPU and 16 GB memory, and each is averaged over at least 20 instances.
%The source code for the experiments can be found in the supplementary material.
We conduct experiments on (i) general instances without any specific underlying structure, and (ii) instances generated based on a road navigation application.
%Due to the space limit, we only present the second set of results and leave results obtained on general instances to the appendix.

\subsection{General Instances}

The first set of instances are generated as follows. The transition probabilities and the initial state distribution are generated uniformly at random (and normalized to ensure that they sum up to $1$). 
We also set an integer parameter $n^*$, and change $n^*$ states to terminal states.
The reward values are first generated uniformly at random from the range $[0,1]$.
Then, we tune the agent's rewards according to a parameter $\beta \in [-1, 1]$, resetting $\widetilde{R}(s, \theta, a) \leftarrow (1 - |\beta|) \cdot \widetilde{R}(s, \theta, a) + \beta \cdot {R}(s, \theta, a)$.
Hence, when $\beta = 0$, the agent's rewards are independent of the principal's; when $\beta = 1$, they are completely aligned; and when $\beta = -1$, they are zero-sum.

%The results are show in Figure~\ref{fig:exp}, where 
We evaluate optimal signaling strategies against a myopic and an AM agent; the latter is equivalent to our threat-based strategy against an FS agent ({\sc Threat-FS}).
We use two benchmarks, which are by nature also the lower and upper bounds of payoffs of other strategies:
i) when the principal cannot send any signal and the agent operates with only the prior knowledge ({\sc NoSig-myop} and {\sc NoSig-AM/FS}; AM and FS agents have the same behavior in this case);
and
ii) when the principal has full control over the agent ({\sc FullControl}).
For ease of comparison, all results are shown as their ratios to results of {\sc FullControl}.
%It appears that patterns of the results do not vary significantly with the discount factors, so we fix them to $0.8$ in the set of results we present. Instead, we run experiments with different numbers of terminate states. 

Figure~\ref{fig:exp} summarizes the results.
It is clearly seen that {\sc OptSig} improves significantly upon {\sc NoSig} in all figures.
%, irrespective of the agent's type.
The gap appears to increase with $\beta$, and when $\beta \ge 0$ (when the agent's rewards are positively correlated to that of the principal), {\sc OptSig} is very closed to {\sc FullControl}. 
%(achieving more than around $90\%$ of the payoff on average). 
It is also noted that differences between results obtained in the myopic setting and in the FS/AM setting are very small (e.g., compare (a) and (b)). This is mainly due to the random nature of the instances: in expectation, future rewards of all actions are the same.
Hence, in the remaining figures we only present results obtained in the myopic setting.
As shown in these figures, payoff improvement offered by the optimal strategies increases slowly with the number of actions and the number of external parameters. Intuitively, as these two numbers increase, the agent's decision making in each state becomes more reliant on advice from the principal.
Nevertheless, the results do not appear to vary insignificantly with other parameters, such as the number of states or the number of terminal states as shown in (e) and (f) (also see %Appendix~\ref{sc:additional_exp} 
the appendix for additional experiment results).

\subsection{Road Navigation Instances}

%We also tested our signaling strategies in a navigation application.
In the navigation application, the agent wants to travel from a starting node to a destination node on a road network, and is free to choose any path.
In each step, the agent picks a road at the current node and travels through it. 
%Hence, the agent's position defines the state. 
The reward the agent receives at each step is a {\em cost} representing the travel time through the chosen road, which depends on the congestion level represented by the external parameter.
The principal, as a social planner, has a preference over the path the agent picks (e.g., in consideration of the overall congestion or noise levels across the city), and this is encoded in a reward function for the principal: whenever the agent picks a road, the principal also receives a cost according to this reward function. 
Naturally, the agent's position (the node the agent is on) defines the state of the MDP. 
For simplicity, we assume that the road network is a directed acyclic graph (DAG), so the agent always reaches the destination in a finite number of steps. 

To generate an instance, we first generate a random DAG with specified numbers of nodes and edges (roads). Let these numbers be $n$ and $m$, respectively ($n \le m \le \frac{n(n-1)}{2}$). We sample a Prüfer sequence of length $n-2$ uniformly at random and then convert it into the corresponding tree. 
We index the nodes according to their order in a breadth-first search.
The node with the smallest/largest index is chosen as the start/destination.
Then we add an edge between a pair of nodes chosen uniformly at random, from the node with the smaller index to the node with the larger index, until there are $m$ edges on the graph. In the case that some node has no outgoing edge and it is not the destination, we also add an edge linking this node to the destination, so the graph generated may actually have more than $m$ edges.
In this way the graph generated is always a DAG.

The results are presented in Figure~\ref{fig:exp_nav}. The results exhibit similar patterns to their counterparts in Figure~\ref{fig:exp}.
%and \ref{fig:additional_exp}.
Nevertheless, the gaps between different strategies appear to be narrower in the FS/AM setting than those in the myopic setting, which is not obvious in Figure~\ref{fig:exp}.

\section{Conclusion}
\label{sec:conclusions}

We described and studied a dynamic model of persuasion in infinite horizon Markov processes.
Our main results characterize the nature and computational complexity of optimal signaling against different types of agents.
A limitation of the current model is that it requires common knowledge of transitions and rewards; hence, studying online versions of our problem \cite{castiglioni2020online} is an immediate future step.
%\paragraph{Broader Impact}
While we focus on the algorithmic aspects of persuasion, our results indicate how a social planner might influence agents optimally. 
In particular implementations, the planner's incentives may not be aligned with societal benefits. In these cases, a careful analysis of the persuasion mechanisms and their moral legitimacy must be considered.

%\clearpage

\section*{Acknowledgments}

This research was sponsored in part by the Deutsche Forschungsgemeinschaft project 389792660 TRR 248--CPEC.
Jiarui Gan was  supported  by  the European Research Council (ERC) under the European Union’s Horizon 2020 research and innovation programme (grant agreement No. 945719).

\clearpage
\onecolumn

\appendix

\section{Omitted Proofs}

\subsection{Proof of Theorem~\ref{thm:harness-optfs}}
\label{sc:proof-hardness}

%\begin{proof}

We show a reduction from the {\sc Maximum Independent Set} problem ({\sc Max-Ind-Set}).
An instance of {\sc Max-Ind-Set} is given by an undirected graph $G= (E,N)$. The goal is to find an {\em independent set} of the maximum size: a set of nodes $N' \subseteq N$ is said to be an independent set if for every pair of nodes $v,u \in N'$ it holds that $\{v, u\} \notin E$.
It is known that {\sc Max-Ind-Set} admits no efficient $(1/m + \epsilon)$-approximation algorithm unless P~=~NP, where $m = |N|$ \cite{zuckerman2006linear}.
The approximation ratios we consider are all multiplicative.
Since the objective value (cumulative reward) might be negative, to make the ratios meaningful, we adjust the objective value by subtracting from it a benchmark value that equals the cumulative reward of the principal when no signal is used. This is a value that can be trivially obtained by the principal and the value of an optimal solution is always at least as large and hence non-negative after the adjustment.

\subsection*{The Reduction}
We only consider the case where $\tilde{\gamma}$ is a constant in $(0, 1/2)$ but note that the reduction easily extends to the case where $\tilde{\gamma} \in [1/2, 1)$. Specifically, we can modify the reduction by inserting a constant number $l = \left \lceil \log_{\tilde{\gamma}}{\frac{1}{2}} \right\rceil$ of dummy states $s'_1, \dots, s'_l$ before each state in the MDP constructed below: in each $s'_\ell$, any action taken results in a deterministic state transition from $s'_\ell$ to $s'_{\ell+1}$ (and from $s'_\ell$ to $s$ when $\ell = l$). 
We also change all occurrence of $\tilde{\gamma}$ in the definition of rewards to $\tilde{\gamma}'$.
This creates an MDP equivalent to the one used in the reduction but with a discount factor $\tilde{\gamma}' < 1/2$. 

Given a {\sc Max-Ind-Set} instance $G=(E,N)$, we construct an MDP illustrated in Figure~\ref{fig:reduction-opt-signaling-FS}, where
\begin{itemize}
\item For each $v\in N$, there are three states $s_v$, $s'_v$, and $s''_v$. In addition, there is a terminal state $s_X$.
\item The initial state is sampled from a uniform distribution over states $s_v$, $v \in N$.
\item For each pair of {\em adjacent} nodes $u$ and $v$, i.e., $\{u,v\}\in E$, there is an action $a_{u,v}$; taking $a_{u,v}$ leads the state transitioning to $s'_u$.
There is no such an action if $u$ and $v$ are not adjacent.
\item The external parameter has two possible values: $\Theta = \{\theta_a, \theta_b\}$. 
\item The agent's reward for each state-action pair is annotated on the corresponding edge in Figure~\ref{fig:reduction-opt-signaling-FS}. It depends on the external parameter only in states $s'_v$, $v \in N$, where the values are given by a function $\rho$ presented on the right. Namely, action $a$ (respectively, $b$) is more profitable when the external parameter is $\theta_a$ (respectively, $\theta_b$). We set the prior distribution of the external parameter to a uniform distribution: $\mu_{s'_v} (\theta_a) = \mu_{s'_v} (\theta_b) = 0.5$.
\end{itemize}
Hence, according to the last point above, the principal can persuade the agent only in states $s'_v$, $v \in N$.
It is not hard to see that through persuasion the principal is able to control the agent's reward between states $s'_v$ and $s''_v$ within the range $[0, 1]$. 
In more detail, reward $0$ corresponds to the case where the principal gives completely uninformative signals (e.g., always sending signal $a$), when the agent can only rely on the prior belief and obtains expected reward $0$; and reward $1$ corresponds to the case where the principal always reveals the true information, when the agent can always pick the correct action to obtain reward $1$.

Finally, we need to specify the principal's reward. 
\begin{itemize}
\item
Let the principal's reward be $1$ only for state-action pairs $(s''_v, b)$, $v\in N$;
for all other state-action pairs we let the reward be $0$.
In other words, the principal's payoff depends on the number of states $s''_v$ at which the agents takes action $b$.
\end{itemize}
%The principal's rewards are also all non-negative in this way to make the multiplicative approximation ratio meaningful.
It can be easily verified that when no signal is used, the agent always prefers to take action $a$ in every state $s_v$, $v\in N$, resulting in cumulative reward $0$ for the principal. Hence, the objective value we consider remains the same after we subtract from it this benchmark cumulative reward.

\definecolor{statefill}{RGB}{204,208,246}

\begin{figure}
\centering
\tikzset{
->, % makes the edges directed
>={Stealth[scale=1.0]}, % makes the arrow heads bold
node distance=2cm, % specifies the minimum distance between two nodes. Change if necessary.
every state/.style={thick, fill=gray!10}, % sets the properties for each ’state’ node
initial text= {$prob.=1/m$}, % sets the text that appears on the start arrow
every edge/.append style={semithick},
initial where = above,
}
\begin{tikzpicture}[baseline={(current bounding box.center)}]
\tikzstyle{every node}=[font=\small] 

% \node[state, initial, initial where=above] (s0) {$s_0$};
% \node[state, below of=s0, yshift=5mm, scale=0.2] (a0) {};

\node[state, initial] (sv) {$s_v$};
\node[state, initial, right of=sv, xshift=20mm] (su) {$s_u$};
\node[left of=sv, xshift=0mm, yshift=2mm] (dotl) {$\dots~\dots$};
\node[right of=sv, xshift=0mm, yshift=2mm] (dotm) {$\dots~\dots$};
\node[right of=su, xshift=0mm, yshift=2mm] (dotr) {$\dots~\dots$};

\node[state, very thick, fill=statefill, below of=sv, yshift=3mm] (sv1) {$s'_v$};
\node[state, very thick, fill=statefill, below of=su, yshift=3mm] (su1) {$s'_u$};
\node[state, below of=sv1, yshift=-5mm] (sv2) {$s''_v$};
\node[state, below of=su1, yshift=-5mm] (su2) {$s''_u$};
%\node[state, below left of=sv1, xshift=-3mm, yshift=3mm] (sv3) {$s_v^*$};
%\node[state, below right of=su1, xshift=3mm, yshift=3mm] (su3) {$s_u^*$};
%\node[state, below left of=a0, xshift=60mm] (sm) {$s_m$};

% \node[state, accepting, left of=sv1, xshift=-3mm] (svX) {};
% \node[state, accepting, right of=su1, xshift=3mm] (suX) {};

\node[state, accepting, below of=sv2, yshift=0mm, xshift=20mm] (sX) {$s_X$};

\draw 
% (s0) edge[left] node{$a, 0$} (a0)
% (a0) edge[dotted, thick, left, bend right] node{$\frac{1}{m}$} (sv)
% (a0) edge[dotted, thick, right, bend left] node{$\frac{1}{m}$} (su)
% (a0) edge[dotted, thick, bend right, left] node{$prob=\frac{1}{m}$} (dotl)
% (a0) edge[dotted, thick, bend left, above] node{$\frac{1}{m}$} (dotr)

(sv) edge[left] node{$b, 0$} (sv1)
(su) edge[right] node{$b, 0$} (su1)
(sv1) edge[out=-150,in=150, left] node{$a, \rho(a,\theta)$} (sv2)
(sv1) edge[left] node{$b$} node[right]{$\rho(b,\theta)$} (sv2)
(su1) edge[right] node{$b$} node[left]{$\rho(b,\theta)$} (su2)
(su1) edge[out=-30,in=30, right] node{$a, \rho(a,\theta)$} (su2)
%(sv1) edge[above] node{$y$} (sv3)
%(su1) edge[above] node{$y$} (su3)

(sv2) edge[out=0,in=180] node[sloped, pos=0.15, below]{$a_{v,u}, 0$}
%node[sloped,pos=0.2,above]{$a_{v,u}, 0$} 
(su1)
(su2) edge[out=180,in=0] node[sloped, pos=0.15, below]{$a_{u,v}, 0$}
%node[sloped,pos=0.2,above]{$a_{u,v}, 0$} 
(sv1)

(sv) edge[out=200,in=190,looseness=1.9] node[pos=0.1,left]{$a, \tilde{\gamma}$} (sX)
(su) edge[out=-20,in=-10,looseness=1.9] node[pos=0.1,right]{$a, \tilde{\gamma}$} (sX)

%(sv3) edge[loop below] node{} (sv3)
%(su3) edge[loop below] node{} (su3)
(sv2) edge[out=-90,in=170, left] node{} node[pos=0.4,left]{$b, \tilde{\gamma}^2$} (sX)
(su2) edge[out=-90,in=10, right] node{} node[pos=0.4,right]{$b, \tilde{\gamma}^2$} (sX)
%(sX) edge[loop below] node{$1\,(0)$} (sX)
;
\end{tikzpicture}
\hspace{-10mm}
\begin{tabular}{ crr } 
\toprule
  & $\theta_a$ & $\theta_b$ \\
\midrule
$a$ & $1$  & $-1$ \\ 
$b$ & $-1$ & $1$ \\
\bottomrule \\[-2mm]
\multicolumn{3}{c}{value of $\rho(\cdot, \cdot)$}
\end{tabular}
\caption{Reduction from {\sc Max-Ind-Set}. 
The two nodes $u$ and $v$ illustrated above are adjacent, i.e., $\{u,v\} \in E$.
All state transitions are deterministic. Annotated on each edge are the action name and the reward of the agent for the corresponding state-action pair. Only the rewards related to states $s'_v$, $v\in N$ depend on the external parameter, whose values are given by $\rho$ presented on the right.
\label{fig:reduction-opt-signaling-FS}}
\end{figure}
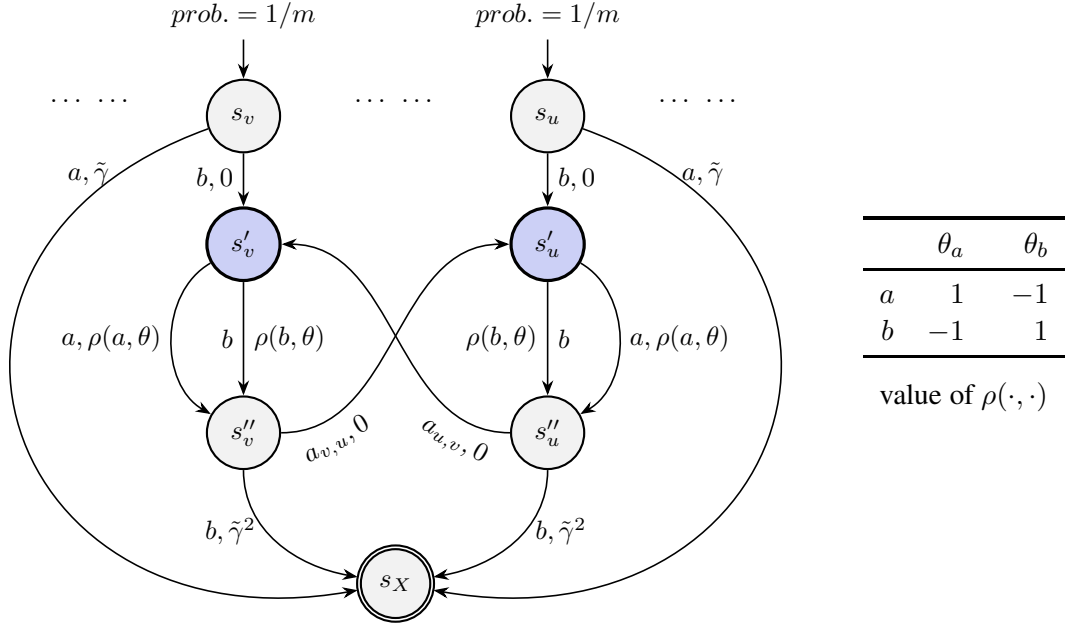
 
\subsection*{Correctness of the Reduction}
We show that there is a size-$k$ independent set on $G$ if and only if there is a signaling strategy that gives the principal payoff $k\cdot \gamma^2/m$ (and the conversion between the independent set and this signaling strategy can be done in polynomial time).
As a result, an efficient approximation algorithm for computing an optimal signaling strategy, if exists, can be efficiently turned into a one for {\sc Max-Ind-Set} while the approximation ratio is preserved. 

\paragraph{The ``only if'' direction.}
Suppose that there is a size-$k$ independent set $N^* \subseteq N$.
We show that the following signaling strategy gives the principal payoff $k\cdot  {\gamma}^2/m$.
\begin{itemize}
\item
In each state $s'_v$, $v\in N^*$, the principal always reveals the true information, recommending the agent to take the correct action with reward $1$. This gives the agent reward $1$ in expectation at that step.
\item
In each state $s'_v$, $v\in N \setminus N^*$, the principal reveals no information, so the agent can only rely on their prior knowledge to obtain an expected reward of $0$.
\end{itemize}

It can be verified that, given that $\tilde{\gamma} < 1/2$, when the above signling strategy is applied, the agent will be (strictly) incentivized to follow the following trajectory at each $s_v$:
\begin{equation*}
%\label{eq:reduction-traj}
\begin{cases}
s_v \to s'_v \to s''_v \to s_X, & \text{ if } v \in N^* \\
s_v \to s_X, & \text{ if } v \in N \setminus N^*
\end{cases}
\end{equation*}
% if $v \in N^*$:
% \[
% s_v \to s'_v \to s''_v \to s_X;
% \]
% and the following trajectory if $v \in N \setminus N^*$:
% \[
% s_v \to s_X.
% \]
This results in payoff $|N^*|\cdot  {\gamma}^2/m = k\cdot {\gamma}^2/m$ for the principal.
Indeed, we argue that by following the above trajectories, the state value function of the agent is as follows; one can easily verify that the actions the agent takes according to the above trajectories are indeed (strictly) optimal with respect these state values.
\begin{itemize}
\item $V(s_X) = 0$.
\item For all $v \in N^*$, we have $V(s''_v) = \tilde{\gamma}^2$, $V(s'_v) = 1 + \tilde{\gamma}^3$, and $V(s_v) = \tilde{\gamma} + \tilde{\gamma}^4$.
\item For all $v \in N \setminus N^*$, we have $V(s'_v) < 1$.
%$V(s''_v) \le \tilde{\gamma} + \tilde{\gamma}^4$, and hence $V(s'_v) \le \tilde{\gamma}^2 + \gamma^5 < 1$ (by assumption $\tilde{\gamma} < 1/2$).
\end{itemize}
In particular, to see that $V(s'_v) < 1$ for every $v \in N \setminus N^*$, even though we did not specify the trajectory after $s'_v$, we can assume that in general the trajectory takes the following form for some $l \in \mathbb{N} \cup \{+\infty\}$:
\[T_l : s'_v \to s''_v \to s'_{u_1} \to s''_{u_1} \to s'_{u_2} \to s''_{u_2} \to \dots \to s'_{u_l} \to s''_{u_l} \to s_X.\] 
Namely, it visits $l$ state pairs $(s'_u, s''_u)$ before it moves to the terminal state.
Observe that in the section between $s'_{u_1}$ and $s''_{u_l}$ in $T_l$, the agent obtains reward $1$ (i.e., reaches some $s'_u$, $u \in N^*$) in at least every four steps given that $N^*$ is an independent set.
Moreover,
since $\tilde{\gamma}^2 < \frac{\tilde{\gamma}^2}{1- \tilde{\gamma}^4} = \tilde{\gamma}^2 + \tilde{\gamma}^6 + \tilde{\gamma}^{10} + \dots$, when the reward $\tilde{\gamma}$ obtained between $s''_{u_l}$ and $s_X$ is at most the cumulative reward the agent would obtain if the trajectory continued after $s_X$ and the agent received reward $1$ in every four steps starting from the second step after $s''_{u_l}$.
Hence, overall, the cumulative reward of $T_l$ is at most that the agent would obtain if they received reward $1$ in every four steps starting from the third step (i.e., after $s'_{u_1}$), that amounts to $\frac{\tilde{\gamma}^2}{1 - \tilde{\gamma}^4} < 1$ (given the assumption that $\tilde{\gamma} < 1/2$).

\paragraph{The ``if'' direction.}
Suppose that some signaling strategy $\pi$ (not necessarily deterministic) gives the principal payoff at least $k\cdot {\gamma}^2/m$. We show that there exists a size-$k$ independent set (and it can be found efficiently).

Let $\sigma$ be an arbitrary optimal policy of the agent in response to $\pi$.
Let $N^* \subseteq N$ contains vertices $v$ such that starting from $s_v$, the agent reaches $s_X$ with a positive probability by following $\sigma$.
Since the principal's cumulative reward is at least $k\cdot {\gamma}^2/m$, it must be that $|N^*| \ge k$.
Note that given $\pi$, both $\sigma$ and $N^*$ can be computed efficiently. We argue that $N^*$ is an independent set to complete the proof.

Indeed, consider any $v \in N^*$.
Since the agent reaches $s_X$ from $s_v$ with a positive probability and $\sigma$ is optimal, the agent must at least weakly prefer action $b$ to any $a_{v,u}$ in state $s''_v$;
it holds for the corresponding Q-function of $\sigma$ that $Q^\sigma(s''_v, b) \ge Q^\sigma(s''_v, a_{v,u})$ for all $u \in N$.
Trivially, we have 
\[Q^(s''_v, b) = \widetilde{R}(s''_v, b) + \tilde{\gamma} \cdot V(s_X) = \tilde{\gamma}^2,\]
and 
\[Q(s''_v, a_{v,u}) = \widetilde{R}(s''_v, a_{v,u}) + \tilde{\gamma} \cdot V(s'_u) = \tilde{\gamma} \cdot V(s'_u).\]
Hence, $V(s'_u) \le \tilde{\gamma}$, which implies that 
\[
Q(s_u, b) 
%= \widetilde{R}(s_u, b) + \tilde{\gamma} \cdot V(s'_u) 
\le \tilde{\gamma}^2 < \tilde{\gamma} = Q(s_u, b).
\]
Consequently, $b$ is strictly worse than $a$ in state $s_u$, so starting from $s_u$ and applying $\sigma$ the agent will reach $s_X$ with probability zero; we have $u \notin N^*$. Since $u$ is an arbitrary vertex adjacent to $v$, it follows immediately that $N^*$ is an independent set.
This completes the proof.

\subsection*{Remarks on the Reduction}

First, the above reduction does not rely on any specific assumption about the tie-breaking behavior of the agent (which defines the policy the agent will choose when there are multiple best responses against $\pi$).

%\paragraph{Applicability to Other Problems}
Second, the reduction can be easily adapted to show the inapproximability of several related problems on reward poisoning (policy teaching) and commitment in stochastic games.
All these problems feature a two-player MDP between a principal and an agent; the players' actions jointly determine the rewards and state transition, and the principal is able to influence the agent's reward in some or all states beforehand (directly, or indirectly through committing to some strategy). The task is to find an optimal way for the principal to influence the agent, which maximizes the principal's discounted cumulative reward.
Indeed, the principal's strategy in the above reduction boils down to setting the overall reward for the agent between each pair of states $s'_v$ and $s''_v$ to a value in $[0,1]$.
To choose the values optimally is an optimal reward poisoning (policy teaching) problem, where the attacker (teacher) has the power to change rewards in these states.
This also corresponds to a stochastic game, where the principal's actions in each state $s'_v$ either give the agent a high payoff $1$ or a low payoff $0$, irrespective of the action the agent plays.
Indeed, it is known that computing an optimal commitment in a stochastic is NP-hard \cite{letchford2012computing} but it was not known how hard it is to approximate.

%\paragraph{Hardness of Computing Finite-memory Strategies}

The proof can also be modified to show the hardness of computing a $k$-memory strategy, simply by inserting $k$ dummy states before each state, so that the principal will always forget what happened in the previous non-dummy state.

% =================================================

\subsection{Proof of Theorem~\ref{thm:varpi-IC}}
\label{sec:pf-varpi-ic}

We show that the value function of $\sigma$ satisfies the Bellman equation.
Specifically, let $V^\sigma: S \times G \to \mathbb{R}$ be the value function of $\sigma$; we will show that the following Bellman equation holds for all $(s,g) \in S \times G$:
\[
\label{eq:V-sigma-bellman}
V^\sigma(s,g) = \max_{b \in A} \left(\widetilde{R}^\varpi((s,g), b) + \tilde{\gamma} \cdot \mathbb{E}_{(s',g') \sim P^\varpi((s,g),b, \cdot)} V^\sigma(s',g')\right).
%\tag{$\star$}
\]
which then implies that $\sigma$ is optimal.

Since $\sigma$ is IC, we have $\sigma(s,g_a) = \mathbf{\hat{e}}_a$. Hence, $V^\sigma$ is the (unique) solution to the following system of equations:
\begin{align}
V^\sigma(s,g_a) 
&= \widetilde{R}^\varpi((s,g_a), a) + \tilde{\gamma} \cdot \mathbb{E}_{(s',g') \sim P^\varpi((s,g_a), a, \cdot)} V^\sigma(s',g') \nonumber \hspace{-15mm}\\
&= \widetilde{R}^\pi((s,g_a), a) + \tilde{\gamma} \cdot \mathbb{E}_{(s',g') \sim P^\pi((s,g_a),a, \cdot)} V^\sigma(s',g'), && \text{for all } s \in S, a \in A
\label{eq:V-sigma-def}
\end{align}
where we replace $\widetilde{R}^\varpi$ and $P^\varpi$ with $\widetilde{R}^\pi$ and $P^\pi$, respectively, according to \eqref{eq:thm-varpi-tR-varpi} and \eqref{eq:thm-varpi-P-varpi-a}.

Consider each $(s,g) \in S \times G$ and the following two cases with respect to $g$.

\paragraph{Case 1.} 
$g \in G_A$. 

Suppose that $g = g_a$ ($a \in A$).
The following key lemma compares $V^\sigma$ and $\overline{V}$. 
Intuitively, $\overline{V}$ serves as a conservative estimate to the expected future reward yielded by $\sigma$.

\begin{lemma}
\label{lmm:V-sigma-tV}
$\mathbb{E}_{\theta \sim \mu_s, g\sim \pi_s(\theta)} V^\sigma(s,g) \ge \overline{V}(s,g_0)$ for all $s \in S$.
%$\mathbb{E}_{(s',g') \sim P^\varpi((s,g_a),a, \cdot)} \left( V^\sigma(s',g') - \overline{V}(s',g_0) \right) \ge 0$ for all $s \in S, a \in A$.
\end{lemma}

\begin{proof}
For ease of description, let $U(s) := \mathbb{E}_{\theta \sim \mu_s, g\sim \pi_s(\theta)} V^\sigma(s,g)$
for all $s \in S$.

According to \eqref{eq:V-sigma-def}, we have
\begin{align}
V^\sigma(s,g_a) 
%&= \widetilde{R}^\pi((s,g_a), \sigma(s,g_a)) + \tilde{\gamma} \cdot \mathbb{E}_{(s',g') \sim P^\pi((s,g_a), \sigma(s,g_a), \cdot)} V^\sigma(s',g') \\
&= \widetilde{R}^\pi((s,g_a), a) + \tilde{\gamma} \cdot \mathbb{E}_{(s',g') \sim P^\pi((s,g_a),a, \cdot)} V^\sigma(s',g'),
\label{eq:V-sigma-def-pi}
\end{align}
for all $s \in S, a \in A$.
The second term on the right side can be rewritten as follows according to \eqref{eq:P-star}, where by writing $\mu_{s'} \cdot \pi_{s'}$ we treat $\mu_{s'} = (\mu_{s'}(\theta))_{\theta \in \Theta}$ as a row vector and $\pi_{s'} = (\pi_{s'}(\theta, g))_{\theta \in \Theta, g \in G_A}$ as a matrix.
\begin{align}
\tilde{\gamma} \cdot \mathbb{E}_{(s',g') \sim P^\pi((s,g_a),a, \cdot)} V^\sigma(s',g')
=\ & \tilde{\gamma} \cdot \mathbb{E}_{s' \sim P(s,a, \cdot), g' \sim \mu_{s'} \cdot \pi_{s'}} V^\sigma(s',g') \nonumber \\
=\ & \tilde{\gamma} \cdot \mathbb{E}_{s' \sim P(s,a, \cdot)} U(s').
\label{eq:V-sigma-def-rewrite-U}
\end{align}
Hence, we have 
\begin{align*}
U(s) 
&= \mathbb{E}_{g_a \sim \mu_s \cdot \pi_s} V^\sigma(s,g_a) \\
&= \mathbb{E}_{g_a \sim \mu_s \cdot \pi_s} \left( \widetilde{R}^\pi((s,g_a), a) + \tilde{\gamma} \cdot \mathbb{E}_{s' \sim P(s,a, \cdot)} U(s') \right) && \text{(plug in \eqref{eq:V-sigma-def-pi} and \eqref{eq:V-sigma-def-rewrite-U})}\\
&= \mathbb{E}_{g_a \sim \mu_s \cdot \pi_s} \widetilde{R}^\pi((s,g_a), a) + \tilde{\gamma} \cdot \mathbb{E}_{g_a \sim \mu_s \cdot \pi_s} \mathbb{E}_{s' \sim P(s,a, \cdot)} U(s').
\end{align*}
We can write the above equation for all $s \in S$ more concisely as
\begin{equation}
\label{eq:U-concise}
U = X + \tilde{\gamma} \cdot T \cdot U,
\end{equation}
where we treat $U = (U(s))_{s\in S}$ and $X = (X(s))_{s\in S}$ as two column vectors, with
$X(s) = \mathbb{E}_{g_a \sim \mu_s \cdot \pi_s} \widetilde{R}^\pi((s,g_a), a)$;
and $T$ is a $|S|$-by-$|S|$ matrix with
\[
T(s, s') = \sum_{\theta \in \Theta} \mu_s(\theta) \sum_{g_a \in G_A} \pi_s(\theta, g_a) \cdot P(s, a, s').
\]

Next, we will show that the following equation holds to complete the proof 
\begin{equation}
\label{eq:V-ge-X-TV}
\overline{V}(\cdot, g_0) \le X + \tilde{\gamma} \cdot T \cdot \overline{V}(\cdot, g_0),
\end{equation}
where $\overline{V}(\cdot, g_0) = \left( \overline{V}(s, g_0) \right)_{s\in S}$ is a column vector.
Indeed, once the above equation holds, we will have 
\[
U-\overline{V}(\cdot, g_0) \ge \tilde{\gamma} \cdot T \cdot \left ( U - \overline{V}(\cdot, g_0)  \right)
\]
by subtracting it from \eqref{eq:U-concise}.
Since all entries of $\tilde{\gamma} \cdot T$ are non-negative, plugging this inequality to itself repeatedly $n$ times gives
\begin{align*}
U-\overline{V}(\cdot, g_0)
&\ge \tilde{\gamma} \cdot T \cdot \left ( U - \overline{V}(\cdot, g_0)  \right) \\
&\ge \tilde{\gamma} \cdot T \cdot \left ( \tilde{\gamma} \cdot T \cdot \left ( U - \overline{V}(\cdot, g_0)  \right)  \right) \\
%&\ge \tilde{\gamma} \cdot T \cdot \left ( \tilde{\gamma} \cdot T \cdot \left ( \tilde{\gamma} \cdot T \cdot \left ( U - \overline{V}(\cdot, g_0)  \right)  \right)  \right) \\
&\dots\\
&\ge \tilde{\gamma}^n \cdot T^n \cdot \left ( U - \overline{V}(\cdot, g_0)  \right).
\end{align*}
Note that all the entries of $T^n$ are in $[0,1]$. Hence, when $n \to \infty$, the right side converges to a zero vector, which implies that $U(s) -\overline{V}(s, g_0) \ge 0$ for all $s \in S$ and hence, the desired result.

Now we show that \eqref{eq:V-ge-X-TV} holds to complete the proof, i.e., 
for all $s \in S$,
\begin{equation}
\overline{V}(s, g_0) \le X(s) + \tilde{\gamma} \cdot \sum_{s'\in S} T(s,s') \cdot \overline{V}(s, g_0).
%\tag{$*$}
\label{eq:V-ge-X-TV-var}
\end{equation}
Expanding the right side, we obtain the following transitions, where the expectation is taken over $g_a \sim \mu_s \cdot \pi_s$: 
\begin{align}
& \mathbb{E} \left[ \widetilde{R}^\pi \left((s,g_{a}),a \right) + \tilde{\gamma} \cdot \mathbb{E}_{s' \sim P(s, a, \cdot)} \overline{V}(s', g_0) \right] \nonumber  \\
=\ & \mathbb{E} \left[ \mathbb{E}_{\theta \sim \Pr(\cdot|g_{a}, \pi_{s})} \widetilde{R}\left(s, \theta, a\right) + \tilde{\gamma} \cdot \mathbb{E}_{s' \sim P(s, a, \cdot)} \overline{V}(s', g_0) \right] && \text{(by \eqref{eq:tR-star})} \nonumber  \\
=\ & \mathbb{E}\left[ \mathbb{E}_{\theta \sim \Pr(\cdot|g_{a}, \pi_{s})} \left(  \widetilde{R}\left(s, \theta, a\right) + \tilde{\gamma} \cdot \mathbb{E}_{s' \sim P(s, a, \cdot)} \overline{V}(s', g_0) \right) \right]\nonumber  \\
%= \text{(\#)},
=\ & \mathbb{E}\left[ \mathbb{E}_{\theta \sim \Pr(\cdot|g_{a}, \pi_{s})} \widetilde{R}^+\left(s, \theta, a\right) \right] && \text{(by \eqref{eq:tR-p})} \nonumber  \\
\ge\ &  \mathbb{E}\left[ \mathbb{E}_{\theta \sim \Pr(\cdot|g_{a}, \pi_{s})} \widetilde{R}^+\left(s, \theta, b\right) \right], \nonumber %\label{eq:T-oV-part-1}
\end{align}
for all $b \in A$.
The last transition holds as $\pi$ is IC with respect to $\widetilde{R}^+$ (so we have \eqref{eq:ic} with $\widetilde{R}^+$ in place of the reward function).
In particular, this holds for  
\[
b^* \in \arg\max_{b \in A} \left( \mathbb{E}_{\theta \sim \mu_s}  \widetilde{R}(s,\theta,b) + \tilde{\gamma} \cdot \mathbb{E}_{s' \sim P(s,b, \cdot)} \overline{V}(s',g_0) \right),
\]
so expanding $\widetilde{R}^+$ according to \eqref{eq:tR-p} gives 
\begin{align}
&\mathbb{E}\left[ \mathbb{E}_{\theta \sim \Pr(\cdot|g_{a}, \pi_{s})} \widetilde{R}^+(s, \theta, b^*) \right] \nonumber \\
=\ & \mathbb{E}\left[ \mathbb{E}_{\theta \sim \Pr(\cdot|g_{a}, \pi_{s})} \left( \widetilde{R}(s, \theta, b^*) + \tilde{\gamma} \cdot \mathbb{E}_{ s' \sim P(s,b^*,\cdot)} \overline{V}(s', g_0) \right) \right] \nonumber \\
=\ & \mathbb{E}\left[ \mathbb{E}_{\theta \sim \Pr(\cdot|g_{a}, \pi_{s})} \widetilde{R}(s, \theta, b^*) \right] + \tilde{\gamma} \cdot \mathbb{E}_{ s' \sim P(s,b^*,\cdot)} \overline{V}(s', g_0).
\label{eq:lmm-V-sigma-tV-star}
%\tag{$*$}
\end{align}
The first term above can be rewritten as follows:
\begin{align*}
\mathbb{E}_{g_{a} \sim \mu_{s} \cdot \pi_{s}}  \mathbb{E}_{\theta \sim \Pr(\cdot|g_{a}, \pi_{s})} \widetilde{R}\left(s, \theta, b^* \right) 
= \mathbb{E}_{\theta\sim \mu_{s}} \widetilde{R}\left(s, \theta, b^* \right)
\end{align*}
as the marginal distribution of $\theta$ is exactly the prior distribution.

Hence, continuing \eqref{eq:lmm-V-sigma-tV-star}, we have
\begin{align*}
%\text{\eqref{eq:lmm-V-sigma-tV-star}} 
&\mathbb{E}\left[ \mathbb{E}_{\theta \sim \Pr(\cdot|g_{a}, \pi_{s})} \widetilde{R}^+(s, \theta, b^*) \right] \\
=\ & \mathbb{E}_{\theta \sim \mu_s}  \widetilde{R}(s,\theta,b^*) + \tilde{\gamma} \cdot \mathbb{E}_{s' \sim P(s,b^*, \cdot)} \overline{V}(s',g_0) \\
=\ & \overline{V}(s, g_0), && \text{(by \eqref{eq:bellman-tV-g0})}
\end{align*}
so \eqref{eq:V-ge-X-TV-var} holds and this completes the proof.
\end{proof}

% \begin{itemize}
% \item 
Since $\sigma(s,g_a) = a$, we have 
\begin{align*}
V^\sigma(s,g_a) 
&= \widetilde{R}^\varpi((s,g_a), a) + \tilde{\gamma} \cdot \mathbb{E}_{(s',g') \sim P^\varpi((s,g_a),a, \cdot)} V^\sigma(s',g') \\
&= \widetilde{R}^\varpi((s,g_a), a) + \tilde{\gamma} \cdot \mathbb{E}_{(s',g') \sim P^\pi((s,g_a),a, \cdot)} V^\sigma(s',g') && \text{(by \eqref{eq:thm-varpi-P-varpi-a})} \\
&= \widetilde{R}^\varpi((s,g_a), a) + \tilde{\gamma} \cdot \mathbb{E}_{s' \sim P(s, a, s')} \mathbb{E}_{\theta' \sim \mu_{s'}, g'\sim \pi_{s'}(\theta') } V^\sigma(s',g') && \text{(by \eqref{eq:P-star})}\\
&\ge \widetilde{R}^\varpi((s,g_a), a) + \tilde{\gamma} \cdot \mathbb{E}_{s' \sim P(s, a, s')} \overline{V}(s',g_0) && \text{(by Lemma~\ref{lmm:V-sigma-tV})}
\label{eq:V-sigma-s-g-mid}
\tag{$*$}
\end{align*}
For ease of description, fix $s$ and $a$ and let 
\[
\Phi(b) := \widetilde{R}^\varpi((s,g_a), b) + \tilde{\gamma} \cdot \mathbb{E}_{s' \sim P(s, b, s')} \overline{V}(s',g_0).
\]
We have 
\begin{align*}
\Phi(b) 
%&= \mathbb{E}_{\theta\sim \Pr(\cdot|g_a,\pi_s)} \widetilde{R}(s,\theta,b) + \tilde{\gamma} \cdot \mathbb{E}_{s' \sim P(s, b, s')} \overline{V}(s',g_0) \\
&= \mathbb{E}_{\theta\sim \Pr(\cdot|g_a,\pi_s)} \left( \widetilde{R}(s,\theta,b) + \tilde{\gamma} \cdot \mathbb{E}_{s' \sim P(s, b, s')} \overline{V}(s',g_0) \right) && \text{(by \eqref{eq:thm-varpi-tR-varpi} and \eqref{eq:tR-star})}\\
&= \mathbb{E}_{\theta\sim \Pr(\cdot|g_a,\pi_s)} \widetilde{R}^+(s,\theta, b).
\end{align*}
By definition, $\pi$ is IC with respect to $\widetilde{R}^+$. Thus, \eqref{eq:ic} holds (with respect to $\widetilde{R}^+$), and this implies that $\Phi(a) \ge \Phi(b)$.
Hence,
\begin{align*}
\text{\eqref{eq:V-sigma-s-g-mid}}
&\ge \widetilde{R}^\varpi((s,g_a), b) + \tilde{\gamma} \cdot \mathbb{E}_{s' \sim P(s, b, s')} \overline{V}(s',g_0) \\
&= \widetilde{R}^\varpi((s,g_a), b) + \tilde{\gamma} \cdot \mathbb{E}_{(s',g') \sim P^\varpi((s,g_a),b, \cdot)} {V}^\sigma(s',g')
\end{align*}
for all $b\in A \setminus \{a\}$, where the last transition is due to the fact that, given signal $g_a$, the state will only transition to the ones in the form $(s', g_0)$ after an action $b\neq a$ is taken.
%(now that $g = g_a$).
Hence, \eqref{eq:V-sigma-bellman} holds.

% \item

\paragraph{Case 2.}
$g = g_0$.

% To show that $\sigma$ is optimal, we proceed by analyzing V-values of the remaining states $S \times \{g_0\}$.  
Note that once the principal starts to send signal $g_0$, the remaining process stays in the subset $S \times \{g_0\}$ of meta-states. The MDP defined on this subset of meta-states is equivalent to $\mathcal{M}^\pinothing$.
%, i.e., the one the agent faces when principal's signaling is completely uninformative.
By definition, $\sigma$ prescribes the same action for each state in this subset as $\bar{\sigma}$ does, so we have $V^\sigma (s,g_0) = \overline{V}(s, g_0)$ for all $s \in S$. Consequently,
\begin{align*}
V^\sigma (s,g_0) & = \overline{V}(s, g_0)
= \max_{b \in A} \left( \mathbb{E}_{\theta \sim \mu_s}  \widetilde{R}(s,\theta,b) + \tilde{\gamma} \cdot \mathbb{E}_{s' \sim P(s,b, \cdot)} \overline{V}(s',g_0) \right) &\text{(by \eqref{eq:bellman-tV-g0})}\\
&= \max_{b \in A} \left( \widetilde{R}^\varpi((s,g_0), b) + \tilde{\gamma} \cdot \mathbb{E}_{(s',g') \sim P^\varpi((s,g),b, \cdot)} \overline{V}(s',g') \right),
%&\ge \widetilde{R}^\varpi((s,g), b) + \tilde{\gamma} \cdot \mathbb{E}_{(s',g') \sim P^\varpi((s,g),b, \cdot)} \overline{V}(s',g')
\end{align*}
where the last transition is due to \eqref{eq:thm-varpi-tR-varpi} and \eqref{eq:thm-varpi-P-varpi-c}; note that $\widetilde{R}^{\varpi}((s,g_0), b) = \widetilde{R}^{\pinothing}((s,g_0), b) = \mathbb{E}_{\theta \sim \mu_s} \widetilde{R}(s,\theta, b)$ as the posterior $\Pr(\cdot|g,\pinothing)$ degenerates to the prior $\mu_s(\cdot)$. 
Hence, \eqref{eq:V-sigma-bellman} holds, too, which completes the proof.
%\end{itemize}

\subsection{Proof of Corollary~\ref{crl:varpi-IC}}

According to Theorem~\ref{thm:varpi-IC}, $\varpi$ incentivizes an FS agent to take all advised actions. Hence, state transition will only happen among the states in $S\times G_A$. Any trajectory generated by using $\varpi$ will be generated with the same probability that it is generated by using $\pi$. The principal obtains the same cumulative reward, as a result.

\section{Generality of IC Action Advices}
\label{sc:IC-action-advice-wlog}

We show that even when the agent is FS, it is without loss of generality to consider only IC action advices. 
This will immediately imply the generality of IC action advices in the myopic setting (which is equivalent to the FS setting with $\tilde{\gamma} = 0$).

Specifically, let $\pi$ be an arbitrary signaling strategy of the principal, and let $\sigma:S\times G \to A$ be a best response of the agent to $\pi$, i.e., an optimal policy in $\mathcal{M}^\pi = \left\langle S\times G, A, P^\pi, \widetilde{R}^\pi \right\rangle$ (defined in Section~\ref{sc:opt-fs}). 
We construct an action advice $\pi^\star$: for each $s, \theta, a$, we set
\[
\pi_s^\star(\theta, g_a) = \sum_{g \in G:\ \sigma(s,g) = a} \pi_s(\theta, g)
\]
(Without loss of generality, we assume that $\pi_s(\theta)$ is supported on a finite set $G$ of signals; if $G$ is infinite, we can change the summation above to integration.)
We prove the following proposition.

\begin{proposition}
Assume that the agent breaks ties by taking the action advised by $\pi^\star$ when there are multiple optimal actions. 
%There exists an IC action advice $\pi^\star$, such that 
Then $\pi^\star$ is IC, and the agent's best response to $\pi^\star$ yields as much cumulative reward in $\mathcal{M}^{\pi^\star} = \left\langle S\times G, A, P^{\pi^\star}, \widetilde{R}^{\pi^\star} \right\rangle$ as $\sigma$ does in $\mathcal{M}^\pi$, both for the agent and the principal.
\end{proposition}

\begin{proof}
%For ease of description, let $\mathcal{M}^\star = \mathcal{M}^{\pi^\star}$ throughout this proof.
We denote by $u(\mathcal{M}^{\pi}, \sigma')$ the agent's payoff (cumulative reward) for executing some policy $\sigma'$ in $\mathcal{M}^{\pi}$.
%; likewise, we use the notation $u(\mathcal{M}^\pi, \sigma')$. 
We first show that $\pi^\star$ is IC, i.e., the following policy $\sigma^\star: S \times G_A \to A$ is an optimal policy in $\mathcal{M}^{\pi^\star}$: $\sigma^\star(s, g_a) = a$ for all $s \in S, a \in A$.
%the agent will be incentivized to always take the action advised by $\pi^\star$.
Note that, by construction, $\pi^\star$ merges signals used by $\pi$, so it is not more informative than $\pi$. Hence, $u(\mathcal{M}^{\pi^\star}, \sigma') \le u(\mathcal{M}^\pi, \sigma)$ for any policy $\sigma'$, so to prove that $\pi^\star$ is IC it suffices to show that $u(\mathcal{M}^{\pi^\star}, \sigma^\star) = u(\mathcal{M}^\pi, \sigma)$.

Indeed, suppose that the agent uses $\sigma^\star$ in $\mathcal{M}^{\pi^\star}$. 
Conditioned on the environment being in state $s$, with probability $\mu_s(\theta) \cdot \pi_s^\star(\theta, g_a)$ the agent takes action $a$ while the realized external parameter is $\theta$.
By construction we have $\pi_s^\star(\theta, g_a) = \sum_{g \in G:\ \sigma(s,g) = a} \pi_s(\theta, g)$. Hence, 
\[
\mu_s(\theta) \cdot \pi_s^\star(\theta, g_a) = \sum_{g \in G:\ \sigma(s,g) = a} \mu_s(\theta) \cdot \pi_s(\theta, g).
\]
The right side is exactly the conditional probability of the pair $(a,\theta)$ when the agent uses $\sigma$ in $\mathcal{M}^\pi$.
As a result, in each state $s \in S$, every pair $(\theta, a) \in \Theta \times A$ arises with the same probability in these two situations. 
Since the state transition and rewards incurred in the original MDP $\mathcal{M}$ depend only on the external parameter and the action taken in each state, the agent's cumulative reward is also the same in expectation in these two situations, and we have $u(\mathcal{M}^{\pi^\star}, \sigma^\star) = u(\mathcal{M}^\pi, \sigma)$.
For the same reason, this also holds for the principal's cumulative reward.

Therefore, $\pi^\star$ is IC, whereby we also showed that the agent's best response policy yields the same cumulative rewards
for both the agent and the principal.
\end{proof}

\section{Dual LP Formulation for \textsc{OptSigMyop}}

\label{sc:opt-myop-dual-LP}
% ===================
% -- Proof by Dual --
% ===================
Note that we can rewrite \eqref{eq:opt-myop-cons-1} as 
\begin{align*}
V(s) \ge \max_{\mathbf{x} \in \mathcal{A}_s} \left[ R^*(s, \mathbf{x}) + \gamma \cdot \sum_{s' \in S} P^*(s, \mathbf{x}, s') \cdot V(s') \right] && \text{for all } s \in S
\end{align*}
Consider the optimization problem on the right side.
%of \eqref{eq:bellman}.
We can replace $\mathbf{x}$ with $\pi$ by using 
\[x(\theta, a) = \mu_s(\theta) \cdot \pi_s(\theta, g_a).\] 
This results in the following LP, where $\pi_s(\theta, a)$ are the variables and the constraints ensure that $\pi$ is an IC action advice.
%(for all $s \in S, \theta \in \Theta, a \in A$).
\begin{subequations}
\label{lp:max-x-in-As}
\begin{align}
\text{maximize:} \quad
&\sum_{\theta \in \Theta} \sum_{a \in A} \left( R(s, \theta, a) + \gamma \sum_{s' \in S} P(s, a, s') \cdot V(s') \right) \cdot \mu_s(\theta) \cdot \pi_s(\theta, g_a) \hspace{-20mm} \tag{\ref{lp:max-x-in-As}} \\
\text{subject to:} \quad
%&x(\theta, a) = \mu_s(\theta) \cdot \pi_s(\theta, g_a) && \text{for all } \theta \in \Theta, a \in  A \label{eq:max-x-in-As-cons-1}\\[2mm]
&\sum_{\theta\in \Theta} \mu_s(\theta) \cdot \pi_s(\theta, g_a) \cdot \widetilde{R}(s,\theta,a) \ge \sum_{\theta\in \Theta} \mu_s(\theta) \cdot \pi_s(\theta, g_a) \cdot \widetilde{R}(s,\theta,b) \hspace{-20mm} \nonumber \\[-2mm]
& && \text{for all } a,b \in A \label{eq:max-x-in-As-cons-2}\\[2mm]
&\sum_{a \in A} \pi_s(\theta, g_a) = 1 && \text{for all } \theta \in \Theta \label{eq:max-x-in-As-cons-3}\\
&\pi_s(\theta, g_a) \ge 0 && \text{for all } \theta \in \Theta,  a \in A \label{eq:max-x-in-As-cons-4}
\end{align}
\end{subequations}
We can further write the above LP in its dual form, which gives the following dual LP, where $I_s(a,b)$, $J_s(\theta)$, and $K_s(a,\theta)$ are the dual variables corresponding to Constraints \eqref{eq:max-x-in-As-cons-2}, \eqref{eq:max-x-in-As-cons-3}, and \eqref{eq:max-x-in-As-cons-4}, respectively.
\begin{subequations}
\label{lp:dual-lp}
\begin{align}
\text{minimize:} \quad
& \sum_{\theta \in \Theta} J_s(\theta) \tag{\ref{lp:dual-lp}} \\
\text{subject to:} \quad
& \left( R(s, \theta, a) + \gamma \sum_{s' \in S} P(s, a, s') \cdot V(s') \right) \cdot \mu_s(\theta) \nonumber \hspace{-10mm} \\
&\qquad = \left( \widetilde{R}(s, \theta, b) - \widetilde{R}(s, \theta, a) \right) \cdot \mu_s(\theta) \cdot I_s(a,b) + J_s(\theta) - K(a, \theta)\hspace{-50mm} \nonumber  \\
& && \text{for all } \theta \in \Theta, a,b \in  A \label{lp:dual-lp-cons-1} \\[2mm]
& I_s(a,b) \ge 0 && \text{for all } a, b \in A \label{lp:dual-lp-cons-2}\\
& K_s(a, \theta) \ge 0 && \text{for all } \theta \in \Theta, a, b \in A \label{lp:dual-lp-cons-3}
\end{align}
\end{subequations}
Since the optimal objective value of the dual LP is equal to that of the original problem, we obtain the following new formulation for \eqref{lp:opt-myop}, where the dual variables are included as additional variables.
\begin{align}
\text{minimize:} \quad &  \sum_{s \in S} z_s \cdot V(s) \label{lp:opt-myop-new}
\\
\text{subject to:} \quad
&V(s) \ge \sum_{\theta \in \Theta} J_s(\theta) && \text{for all } s \in S \tag{\ref{lp:opt-myop-new}a}
\label{eq:opt-myop-new-cons-1} \\
&\text{\eqref{lp:dual-lp-cons-1}--\eqref{lp:dual-lp-cons-3}} && \text{for all } s \in S \tag{\ref{lp:opt-myop-new}b}
\label{eq:opt-myop-new-cons-2} 
\end{align}
Note that in Constraint~\eqref{eq:opt-myop-new-cons-1} we do not require that the right side is minimized. 
%This will be implicitly enforced by the objective of minimizing $\sum_{s \in S} z_s \cdot V(s)$.
Indeed, for any feasible solution $(V, I_s, J_s, K_s)$ to \eqref{lp:opt-myop-new}, we have 
\[
V(s) \ge \sum_{\theta \in \Theta} J_s(\theta) \ge 
%\min_{\substack{I'_s, J'_s, k'_s\\[0.5mm]\text{s.t. \eqref{lp:dual-lp-cons-1}--\eqref{lp:dual-lp-cons-3}}}}\ \sum_{\theta \in \Theta} J'_s(\theta) 
\textsc{Opt-Dual}
= \max_{\mathbf{x} \in \mathcal{A}_s} \left[ R^*(s, \mathbf{x}) + \gamma \cdot \sum_{s' \in S} P^*(s, \mathbf{x}, s') \cdot V(s') \right],
\]
so $V$ is a feasible solution to the original LP \eqref{lp:opt-myop}, where \textsc{Opt-Dual} denotes the optimal value of \eqref{lp:dual-lp}.
Conversely, any feasible solution $V'$ to \eqref{lp:opt-myop}, corresponds to a feasible solution $(V', I'_s, J'_s, K'_s)$ to \eqref{lp:opt-myop-new}, where $(I'_s, J'_s, K'_s)$ is a solution to \eqref{lp:dual-lp} with $V$ in the coefficients fixed to $V'$; namely, we have
\[
V'(s) \ge \max_{\mathbf{x} \in \mathcal{A}_s} \left[ R^*(s, \mathbf{x}) + \gamma \cdot \sum_{s' \in S} P^*(s, \mathbf{x}, s') \cdot V'(s') \right] = \textsc{Opt-Dual} = \sum_{\theta \in \Theta} J'_s(\theta).
\]

% gamma_pr = gamma_ag = 0.8, n_termin = 0, Theta = 10, A = 10, S = 2, 4, ..., 20

% gamma_pr = gamma_ag = 0.8, n_termin = 0, Theta = 10, A = 10, S = 2, 4, ..., 20

% gamma_pr=0.1,0.2,...,1.0, gamma_ag=0.8, n_termin = 5, Theta = 10, A = 10, S = 10

% gamma_pr = 0.8, gamma_ag=0.1,0.2,...,1.0, n_termin = 5, Theta = 10, A = 10, S = 10

\begin{NoHyper}
\begin{figure}
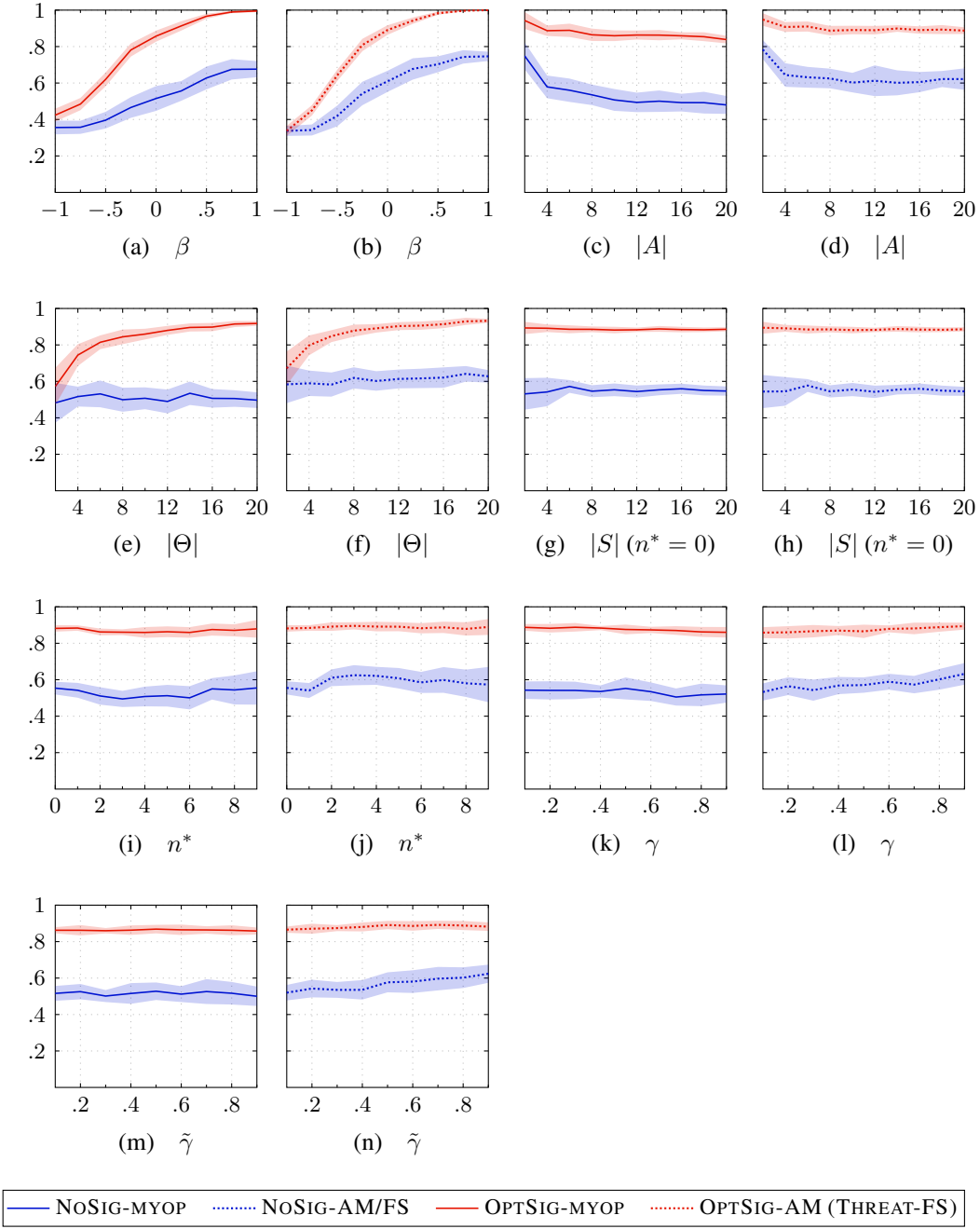

\centering
% \begin{tabular}{l@{\hspace{-7mm}}l@{\hspace{-7mm}}l@{\hspace{-7mm}}l@{\hspace{-7mm}}l}
\begin{tabular}{l@{\hspace{-9mm}}l@{\hspace{-9mm}}l@{\hspace{-9mm}}l}
% ========================
\begin{myplot}
    {
        title={(a)\quad $\beta$},
        %xlabel={$\beta$},
        %yticklabels={\empty},
        %legend entries={ {\sc NoSig-myop}, {\sc NoSig-AM/FS}, {\sc OptSig-myop}, {\sc OptSig-AM (Threat-FS)} },%
 	    %legend to name=leg:all,
    }
    \addplotmyop{random_termin_5.dat}
\end{myplot}%\\[.5mm]
&
% ========================
\begin{myplot}
    {
        title={(b)\quad $\beta$},
        yticklabels={\empty},
    }
    \addplotfs{random_termin_5.dat}
\end{myplot}%
&
% ========================
\begin{myplot}
    {
        title={(c)\quad $|A|$},
        yticklabels={\empty},
        xmin=2, xmax=20,
        xtick={ 4, 8, 12, 16, 20},
        xticklabels={$4$,$8$,$12$,$16$,$20$},
    }
    \addplotmyop{random_var_A.dat}
\end{myplot}%\\[.5mm]
&
% ========================
\begin{myplot}
    {
        title={(d)\quad $|A|$},
        yticklabels={\empty},
        xmin=2, xmax=20,
        xtick={ 4, 8, 12, 16, 20},
        xticklabels={$4$,$8$,$12$,$16$,$20$},
    }
    \addplotfs{random_var_A.dat}
\end{myplot}%\\[.5mm]
\\[2mm]
% ========================
\begin{myplot}
    {
        title={(e)\quad $|\Theta|$},
        xmin=2, xmax=20,
        xtick={ 4, 8, 12, 16, 20},
        xticklabels={$4$,$8$,$12$,$16$,$20$},
    }
    \addplotmyop{random_var_Theta.dat}
\end{myplot}%\\[2mm]
&
% ========================
\begin{myplot}
    {
        title={(f)\quad $|\Theta|$},
        yticklabels={\empty},
        xmin=2, xmax=20,
        xtick={ 4, 8, 12, 16, 20},
        xticklabels={$4$,$8$,$12$,$16$,$20$},
    }
    \addplotfs{random_var_Theta.dat}
\end{myplot}
&
% ========================
\begin{myplot}
    {
        title={(g)\quad $|S|$ ($n^*=0$)},
        yticklabels={\empty},
        xmin=2, xmax=20,
        xtick={ 4, 8, 12, 16, 20},
        xticklabels={$4$,$8$,$12$,$16$,$20$},
    }
    \addplotmyop{random_var_S.dat}
\end{myplot}
&
% ========================
\begin{myplot}
    {
        title={(h)\quad $|S|$ ($n^*=0$)},
        yticklabels={\empty},
        xmin=2, xmax=20,
        xtick={ 4, 8, 12, 16, 20},
        xticklabels={$4$,$8$,$12$,$16$,$20$},
    }
    \addplotfs{random_var_S.dat}
\end{myplot}
\\[2mm]
% ========================
\begin{myplot}
    {
        title={(i)\quad $n^*$},
        xmin=0, xmax=9,
        xtick={0, 2, 4, 6, 8},
        xticklabels={$0$,$2$,$4$,$6$,$8$},
    }
    \addplotmyop{random_var_n_termin.dat}
\end{myplot}%\\[.5mm]
&
% ========================
\begin{myplot}
    {
        title={(j)\quad $n^*$},
        yticklabels={\empty},
        xmin=0, xmax=9,
        xtick={0, 2, 4, 6, 8},
        xticklabels={$0$,$2$,$4$,$6$,$8$},
    }
    \addplotfs{random_var_n_termin.dat}
\end{myplot}%\\[2mm]
&
% ========================
\begin{myplot}
    {
        title={(k)\quad $\gamma$},
        yticklabels={\empty},
        xmin=0.1, xmax=0.9,
        xtick={ .2, .4, .6, .8, 1.0},
        xticklabels={$.2$,$.4$,$.6$,$.8$, $1.0$},
    }
    \addplotmyop{random_var_gamma_pr.dat}
\end{myplot}
&
% ========================
\begin{myplot}
    {
        title={(l)\quad $\gamma$},
        yticklabels={\empty},
        xmin=0.1, xmax=0.9,
        xtick={ .2, .4, .6, .8, 1.0},
        xticklabels={$.2$,$.4$,$.6$,$.8$, $1.0$},
    }
    \addplotfs{random_var_gamma_pr.dat}
\end{myplot}
\\[2mm]
% ========================
\begin{myplot}
    {
        title={(m)\quad $\tilde{\gamma}$},
        xmin=0.1, xmax=0.9,
        xtick={ .2, .4, .6, .8, 1.0},
        xticklabels={$.2$,$.4$,$.6$,$.8$, $1.0$},
    }
    \addplotmyop{random_var_gamma_ag.dat}
\end{myplot}
&
% ========================
\begin{myplot}
    {
        title={(n)\quad $\tilde{\gamma}$},
        yticklabels={\empty},
        xmin=0.1, xmax=0.9,
        xtick={ .2, .4, .6, .8, 1.0},
        xticklabels={$.2$,$.4$,$.6$,$.8$, $1.0$},
    }
    \addplotfs{random_var_gamma_ag.dat}
\end{myplot}\\[2mm]
\end{tabular}
\ref{leg:all}
\caption{\small Comparison of signaling strategies: a full set of results including those in Figure~\ref{fig:exp} and additional ones. All results are shown as ratios to {\sc FullControl} on the y-axes. Meanings of x-axes are noted in the captions. Shaded areas represent standard deviations (mean $\pm$ standard deviation). 
In all figures, we fix $|S| = |\Theta| = |A| = 10$, $\gamma = \tilde{\gamma} = 0.8$, $n^*=5$, and $\beta=0$ unless they are variables.}
\label{fig:additional_exp}
\end{figure}
\end{NoHyper}

\begin{NoHyper}
\begin{figure*}
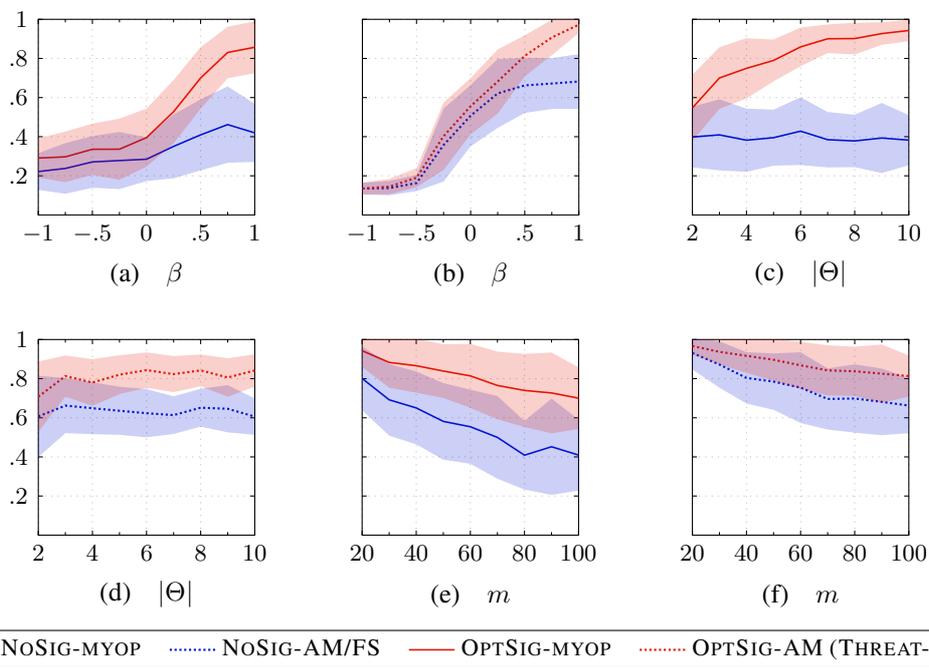

\centering
\begin{tabular}
% {@{\hspace{-9mm}}l@{\hspace{-8mm}}l@{\hspace{-9mm}}l@{\hspace{-9mm}}l@{\hspace{-9mm}}l@{\hspace{-9mm}}l}
{@{\hspace{0mm}}l@{\hspace{0mm}}l@{\hspace{0mm}}l}
% ========================
\begin{myplot}
    {
        title={(a)\quad $\beta$},
    }
    \addplotmyop{road_nav_beta_sorted.dat}
\end{myplot}%\\[.5mm]
&
\begin{myplot}
    {
        title={(b)\quad $\beta$},
        yticklabels={\empty},
    }
    \addplotfs{road_nav_beta_sorted.dat}
\end{myplot}
&
\begin{myplot}
    {
        title={(c)\quad $|\Theta|$},
        yticklabels={\empty},
        xmin=2, xmax=10,
        xtick={ 2, 4, 6, 8, 10},
        xticklabels={$2$, $4$, $6$, $8$, $10$},
    }
    \addplotmyop{road_nav_Theta_sorted.dat}
\end{myplot}
\\[2mm]
\begin{myplot}
    {
        title={(d)\quad $|\Theta|$},
        xmin=2, xmax=10,
        xtick={ 2, 4, 6, 8, 10},
        xticklabels={$2$, $4$, $6$, $8$, $10$},
    }
    \addplotfs{road_nav_Theta_sorted.dat}
\end{myplot}
&
\begin{myplot}
    {
        title={(e)\quad $m$},
        yticklabels={\empty},
        xmin=20, xmax=100,
        xtick={ 20, 40, 60, 80, 100},
        xticklabels={$20$, $40$, $60$, $80$, $100$},
    }
    \addplotmyop{road_nav_road_sorted.dat}
\end{myplot}
&
\begin{myplot}
    {
        title={(f)\quad $m$},
        yticklabels={\empty},
        xmin=20, xmax=100,
        xtick={ 20, 40, 60, 80, 100},
        xticklabels={$20$, $40$, $60$, $80$, $100$},
    }
    \addplotfs{road_nav_road_sorted.dat}
\end{myplot}
% ========================
\end{tabular}
\ref{leg:all}
\caption{\small Comparison of signaling strategies in a navigation application, with {\em a uniform congestion level} at each step. All other settings are the same as Figure~\ref{fig:exp_nav}.
}
\label{fig:additional_exp_nav_sorted}
\end{figure*}
\end{NoHyper}

\section{Additional Experiment Results}
\label{sc:additional_exp}

%\paragraph{Results Omitted in Figure~\ref{fig:exp}}

A full set of results obtained on general instances, including those presented in Figure~\ref{fig:exp} is shown in Figure~\ref{fig:additional_exp}.
For the navigation application, we also present an additional set of results in \ref{fig:additional_exp_nav_sorted}.
These results are obtained on instances with uniform congestion levels, i.e., all roads have the same congestion level at each step (however, the costs for the roads may still be different).
Despite this difference, the results exhibit very similar patterns to those shown in Figure~\ref{fig:exp_nav}.

\end{document}